\documentclass[11pt]{article}
\usepackage[utf8]{inputenc}
\usepackage[margin=1in]{geometry} 
\usepackage[usenames,x11names]{xcolor}

\usepackage{amsmath,amssymb,amsthm}
\usepackage{mathtools}
\usepackage{thm-restate}

\usepackage{enumerate} 

\usepackage[linesnumbered,ruled,vlined]{algorithm2e}

\usepackage[colorlinks,linkcolor=DeepSkyBlue3,citecolor=SpringGreen3]{hyperref}
\usepackage[nameinlink,capitalise]{cleveref} 

\usepackage{float}

\usepackage{xspace}

\newtheorem{theorem}{Theorem}
\newtheorem{corollary}[theorem]{Corollary}
\newtheorem{lemma}[theorem]{Lemma}

\newtheorem{observation}{Observation}

\newtheorem{proposition}[theorem]{Proposition}
\theoremstyle{definition}
\newtheorem{definition}[theorem]{Definition}
\newtheorem*{remark}{Remark}
\theoremstyle{plain}

\crefname{question}{Question}{Questions}

\DeclareMathOperator{\E}{\mathrm{E}}

\DeclareMathOperator{\poly}{poly}

\DeclareMathOperator*{\argmax}{arg\,max}

\newcommand{\cD}{\ensuremath{\mathcal{D}}}

\newcommand{\cQ}{\ensuremath{\mathcal{Q}}}

\newcommand{\cS}{\ensuremath{\mathcal{S}}}

\newcommand{\qkCol}{\ensuremath{(q,k)}\textnormal{-\textsc{Colorability}}\xspace}

\newcommand{\rhoIndepSet}{\ensuremath{\rho}\textnormal{-\textsc{IndepSet}}\xspace}
\newcommand{\qkSAT}{\ensuremath{(q,k)}\textnormal{-\textsc{SAT}}\xspace}
\newcommand{\qtwokSAT}{\ensuremath{(2,k)}\textnormal{-\textsc{SAT}}\xspace}

\newcommand{\kSHPP}{\ensuremath{k}\textnormal{-\textsc{SHPP}}\xspace}

\title{New Graph and Hypergraph Container Lemmas with Applications in Property Testing}

\author{Eric Blais \and Cameron Seth}

\begin{document}
\maketitle

\begin{abstract}
The graph and hypergraph container methods are powerful tools with a wide range of applications across combinatorics.
Recently, Blais and Seth (FOCS 2023) showed that the graph container method is particularly well-suited for the analysis of the natural canonical tester for two fundamental graph properties: having a large independent set and $k$-colorability.
In this work, we show that the connection between the container method and property testing extends further along two different directions.

First, we show that the container method can be used to analyze the canonical tester for many other properties of graphs and hypergraphs.
We introduce a new hypergraph container lemma and use it to give an upper bound of $\widetilde{O}(kq^3/\epsilon)$ on the sample complexity of $\epsilon$-testing satisfiability, where $q$ is the number of variables per constraint and $k$ is the size of the alphabet.
This is the first upper bound for the problem that is polynomial in all of $k$, $q$ and $1/\epsilon$.
As a corollary, we get new upper bounds on the sample complexity of the canonical testers for hypergraph colorability and for every semi-homogeneous graph partition property.

Second, we show that the container method can also be used to study the query complexity of (non-canonical) graph property testers.
This result is obtained by introducing a new container lemma for the class of all \emph{independent set stars}, a strict superset of the class of all independent sets.
We use this container lemma to give a new upper bound of $\widetilde{O}(\rho^5/\epsilon^{7/2})$ on the query complexity of $\epsilon$-testing the $\rho$-independent set property.
This establishes for the first time the non-optimality of the canonical tester for a non-homogeneous graph partition property.
\end{abstract}

\section{Introduction}
Many combinatorial problems can be tackled with a brute-force approach that proceeds by enumerating all the independent sets in a graph or a hypergraph.
However, this brute-force approach is generally of limited utility because a graph on $n$ vertices may have up to $2^n$ independent sets.
The \emph{graph} and \emph{hypergraph container methods} are general tools that circumvent this limitation in many settings.
In particular, for many natural classes of graphs, such as regular graphs, nearly regular graphs, random graphs, or graphs that are sufficiently dense, graph container lemmas show that for any graph $G=(V,E)$ in the class, there exists a relatively small collection of \emph{containers}, each of which is a subset of $V$ and of size much smaller than $|V|$, such that every independent set is contained in at least one of these containers.

The graph container method was originally introduced by Kleitman and Winston~\cite{kleitmanWinston1982} to bound the number of square-free graphs.
Later, Sapozhenko (see~\cite{Sapozhenko05} and the references therein) showed that the graph container method can also be used to resolve many other enumeration problems.
More recently, Balogh, Morris, and Samotij~\cite{baloghIndependentSetsHypergraphs2015} and Saxton and Thomason~\cite{saxtonThomasonHypergraphContainers2015} extended the container method to hypergraphs.
The resulting hypergraph container method has since been used to obtain a number of striking results throughout combinatorics.
(See the excellent survey~\cite{baloghMorrisSamotijContainersSurvey2018} for an introduction to the method and its various applications.)
The graph and hypergraph container methods have also recently found uses in algorithmic settings for counting independent sets in bipartite graphs~\cite{JenssenPP23} and for obtaining faster exact algorithms for many NP-complete problems on almost-regular graphs~\cite{Zamir23}.

Recently, we~\cite{blaisSeth} showed that there is also a connection between the graph container method and graph property testing.
The connection was shown via canonical property testers, which are bounded-error randomized algorithms that aim to distinguish between graphs having some property and graphs that are $\epsilon$-far from the property by sampling a uniformly random set of vertices $S$ and only inspecting the induced subgraph $G[S]$.
(See \cref{sect:prelim} for the formal definitions.)
We used the graph container method to show new upper bounds on the number of vertices that must be sampled by canonical testers for two classic graph properties: having a large independent set (or a large clique) and $k$-colorability.

The results in~\cite{blaisSeth} raise the question of whether the container method has other applications in property testing.
In particular, can the graph or hypergraph container method be used to obtain strong bounds on the sample complexity of the canonical tester for other property testing problems?
In another direction, the results of \cite{blaisSeth} only use the container method to analyze canonical testers.
In general, however, we are interested in non-canonical testers that can query individual (potential) edges to determine if they are in the graph or not. 
Can the graph container method also be used to analyze non-canonical testers, where the measure of interest is the query complexity instead of the sample complexity?

\subsection{Results}

\paragraph{Testing satisfiability and related properties.}
In our first result we establish a new hypergraph container lemma and show that it can be used to obtain new bounds on the sample complexity for the problem of \emph{satisfiability testing}.

In the satisfiability testing problem, which we denote by \qkSAT, the input is a constraint satisfaction problem $\phi$ that contains $n$ constraints, each of which contains $q$ variables that take values from an alphabet of size $k$.
The task of a tester is to distinguish the case where $\phi$ is satisfiable and the case where at least $\epsilon \binom{n}{q}$ constraints must be removed from $\phi$ to make it satisfiable.
The canonical tester samples a set $S$ of variables and accepts if the restriction of $\phi$ to the set $S$ of variables is satisfiable.
This tester has \emph{one-sided error}, in that it always accepts when $\phi$ is satisfiable.
We show that the canonical tester for satisfiability is efficient in a strong sense.

\begin{theorem}
\label{thm:sat}
The sample complexity of $\epsilon$-testing the \qkSAT property is $\widetilde{O}(\frac{k q^3}{\epsilon})$.\footnote{Here and throughout the article, we use $\widetilde{O}(\cdot)$ and $\widetilde{\Omega}(\cdot)$ notation to hide terms that are polylogarithmic in the argument. See \Cref{sect:testing-sat,sect:testing-clique-queries} for the precise formulation of the main theorems.}
\end{theorem}

Satisfiability testing was first studied by Alon and Shapira~\cite{alon2003testing} who showed that the sample complexity of the canonical tester for $(q,k)$-SAT is independent of $n$ and polynomial in $1/\epsilon$.
Specifically, they showed that its sample complexity is roughly $2^{k^{2q}} / \epsilon^2$.
Sohler~\cite{Sohler12} improved the bounds to show that the sample complexity of the canonical tester is roughly $k^{3q} / \epsilon$, and in particular has only a singly exponential dependence on $q$ and a linear dependence on $1/\epsilon$.
\Cref{thm:sat} is the first result to show that the sample complexity for satisfiability testing is polynomial in all three of $q$, $k$, and $1/\epsilon$, and that the dependence on $k$ is linear for all possible values of $q$.
Note that the linear dependence on $1/\epsilon$ in the theorem is optimal, but we do not know whether the bound is optimal in terms of $k$ or $q$ or not.

Sample complexity bounds for testing satisfiability can be used to obtain bounds on the sample complexity for testing $k$-colorability of hypergraphs as well~\cite{alon2003testing,Sohler12}.
In this problem, denoted \qkCol, the input is a $q$-uniform hypergraph $H$ on $n$ vertices and the tester must distinguish the case where $H$ is $k$-colorable from the one where at least $\epsilon \binom{n}{q}$ edges need to be removed to make $H$ $k$-colorable.
The sample complexity bound for testing satisfiability implies a corresponding bound for the sample complexity of the natural canonical tester for $k$-colorability.

\begin{corollary}
\label{cor:k-col}
The sample complexity of $\epsilon$-testing the \qkCol property is $\widetilde{O}(\frac{kq^3}{\epsilon})$.
\end{corollary}

For $q=2$, the bound in \cref{cor:k-col} recovers the best-known sample complexity bound for testing $k$-colorability of graphs from~\cite{blaisSeth}, but via a different argument that, as far as we can tell, uses the graph container method in fundamentally different ways.
Furthermore, \cref{cor:k-col} extends the colorability testing result of~\cite{blaisSeth} to show that $\widetilde{O}(k/\epsilon)$ samples also suffice to test $k$-colorability of $q$-hypergraphs for any constant $q$.
In the general setting, Alon and Shapira~\cite{alon2003testing} showed that the sample complexity of the canonical $k$-colorability tester is at most $(qk)^{O(q)}/\epsilon^2$.\footnote{The prior results on testing colorability of hypergraphs use a different measure of distance which give different sample complexity bounds, but these are converted to use our definition here. See the remark in \cref{sect:prelim} for more details.}
This bound was improved to $(qk)^{O(q)}/\epsilon$ by Sohler~\cite{Sohler12} and to $q^{O(q)} k^2/\epsilon^2$ by Czumaj and Sohler~\cite{czumaj2001testing}.
The bound in \cref{cor:k-col} improves on both of these incomparable results and is the first to show that $k$-colorability testing has polynomial dependence on $q$ and linear dependence on $k.$

\Cref{thm:sat} also implies new sample complexity bounds for testing all \emph{semi-homogeneous partition properties}, a class of graph properties first studied by Nakar and Ron~\cite{NakarRon2018} (where they are called \emph{0--1 graph partition properties}) and by Fiat and Ron~\cite{fiat2021efficient}.
This is the class of graph partition properties where between two parts the edge density is either $0,1$ or has no restrictions.
(See \cref{sect:prelim} for the formal definitions.)
This is a natural class of graph properties, which includes $k$-colorability, being a biclique, and many other properties as well. 
This class is also interesting in its own right, as it has connections to efficient distance approximation~\cite{fiat2021efficient} and it is exactly the class of graph partition properties that can be tested with $\poly(1/\epsilon)$ queries with one-sided error \cite{NakarRon2018}.
Let \kSHPP denote the set of semi-homogeneous partition properties with exactly $k$ parts. 
\Cref{thm:sat} implies the following bound on the sample complexity of the natural canonical tester for semi-homogeneous partition properties.

\begin{corollary}
\label{cor:01-partition}
The sample complexity of $\epsilon$-testing any property in \kSHPP is $\widetilde{O}\left(\frac{k}{\epsilon}\right)$.
\end{corollary}

The bound in \cref{cor:01-partition} again recovers the best-known bound on the sample complexity for testing $k$-colorability of graphs.
In general, it improves on both of the incomparable previous best bounds on the order of $k \log k /\epsilon^2$~\cite{NakarRon2018} and $k^6/\epsilon$~\cite{Sohler12} on the sample complexity for testing \kSHPP properties with one-sided error.

\paragraph{Query complexity of testing independent sets.}
In our second result we use the graph container method to analyze the query complexity of a non-canonical tester for the \rhoIndepSet property, which is the set of all graphs on $n$ vertices that contain an independent set on $\rho n$ vertices.
Note that testing the \rhoIndepSet property is equivalent to testing the \textsc{$\rho$-Clique} property by considering the complement of the graph.

The optimal canonical $\epsilon$-tester for the \rhoIndepSet property has sample complexity $\widetilde{\Theta}(\rho^3/\epsilon^2)$~\cite{blaisSeth,feigeCliqueTesting2004}.
The query complexity of the canonical tester is simply the square of its sample complexity, $\widetilde{\Theta}(\rho^6/\epsilon^4).$
In particular, when $\rho = \frac12$ and $\epsilon = 1/\sqrt{n}$, the canonical tester must sample a constant fraction of the graph, and so it does not have sublinear query complexity.
Blais and Seth~\cite{blaisSeth} asked specifically whether there exists a general tester with sublinear query complexity for this regime.

Our next result shows that there is a non-adaptive $\epsilon$-tester for $\rho$-independent sets with query complexity that is polynomially smaller than that of the canonical tester and, in particular, that can $\frac1{\sqrt{n}}$-test the $\frac12$-independent set property with sublinear query complexity.

\begin{theorem}
\label{thm:cliques-queries}
The query complexity of $\epsilon$-testing the \rhoIndepSet property is $\widetilde{O}(\frac{\rho^5}{\epsilon^{7/2}})$.
\end{theorem}

By definition, the problem of $\epsilon$-testing the \rhoIndepSet property is non-trivial only when $\epsilon < \rho^2$, as otherwise there are no graphs that are $\epsilon$-far from having a $\rho$-independent set.
When $\epsilon = \Theta(\rho^2)$, the query complexity of both the canonical tester and the non-adaptive tester in the proof of \cref{thm:cliques-queries} is $O(1/\epsilon)$, which is optimal.
And in all other non-trivial regimes where $\epsilon = o(\rho^2)$, the bound in \cref{thm:cliques-queries} is strictly better than the query complexity bound for the canonical tester. 

It has been known for some time that the gap between the query complexity of the canonical tester and the optimal tester for any graph property is at most quadratic~\cite{GoldreichTrevisan03}, and this holds even for adaptive testers, which are testers that can determine which edge to query next based on what the previous edge queries have revealed about the graph $G$.
But in general there are few graph properties for which we know whether the canonical tester has optimal query complexity or not.
\Cref{thm:cliques-queries} show that the canonical tester for the \rhoIndepSet property does not have optimal query complexity.
Further, since our tester is non-adaptive, \Cref{thm:cliques-queries} also shows that the \rhoIndepSet property is one of the few graph properties for which we know that the canonical tester does not have the same query complexity the best non-adaptive tester.
Previously, a similar separation was known for the property of being an empty graph, a collection of isolated cliques~\cite{GoldreichRon11}, or a blow-up of a fixed graph~\cite{AvigadGoldreich11}, all of which fit into the class of homogeneous graph partition properties.

Note that the best lower bounds for the query complexity of the \rhoIndepSet property are the incomparable lower bounds of $\Omega(1/\epsilon)$ queries~\cite{GoldreichRon11} and the lower bound of $\Omega(\rho^3/\epsilon^2)$ that follows from the lower bound on the sample complexity of the canonical tester~\cite{feigeCliqueTesting2004}.

\subsection{Techniques: New Graph and Hypergraph Container Lemmas}

The analysis of the testers that establish \Cref{thm:sat,thm:cliques-queries} both require new container lemmas. 
Before we introduce these lemmas, let us briefly describe the graph container lemma introduced in~\cite{blaisSeth} to establish the sample complexity of the canonical tester for the \rhoIndepSet property.

\begin{lemma}[Lemma 5 in~\cite{blaisSeth}]
\label{lem:BSindepset}
Let $G = (V,E)$ be a graph on $n$ vertices that is $\epsilon$-far from the \rhoIndepSet property.
Then there is a collection of fingerprints $\mathcal{F} \subseteq P(V)$ and a function $C: \mathcal{F} \rightarrow P(V),$\footnote{$P(V)$ denotes the power set of $V.$} which we call the container function, such that for every independent set $I$ in $G$, there exists a $F \in \mathcal{F}$ that satisfies 
$F \subseteq I \subseteq C(F)$,
$|F| \le \widetilde{O}(\frac{\rho^2}{\epsilon})$, and
\[
|C(F)| \le \left( \rho - |F| \cdot \widetilde{\Omega}\left( \tfrac{\epsilon}{\rho} \right) \right) n.
\]
\end{lemma}

Informally, \cref{lem:BSindepset} says that there is a small collection of containers, that being $\mathcal{C}=\{C(F) : F \in \mathcal{F}\},$ such that every independent set is a subset of at least one of the containers in the collection.
This follows from the fact that there can be at most $n^{\widetilde{O}(\rho^2/\epsilon)}$ fingerprints of cardinality $\widetilde{O}(\rho^2/\epsilon)$.
Furthermore, the lemma also says that each container in the collection is small; they all have cardinality at most $(\rho - \widetilde{\Omega}(\epsilon/\rho))n$, a bound 
that is noticeably smaller than the cardinality $\rho n$ of the independent sets that are present in all graphs that have the property \rhoIndepSet.

But in fact \cref{lem:BSindepset} says even more: it gives a tight trade-off between the size of a container and that of its fingerprint.
The larger the container, the stronger the guarantee on the size of its fingerprint.
(Or, equivalently, the larger the fingerprint, the smaller is its associated container.)
For instance, it guarantees that there can be at most $n^{\widetilde{O}(1)}$ containers of cardinality $(\rho - \widetilde{\Theta}(\epsilon/\rho))n$ in the collection, and that if the collection contains $n^{\widetilde{\Theta}(\rho^2/\epsilon)}$ containers, most of them have cardinality $c \cdot \rho n$ for some constant $c < 1$.
This trade-off is critical in obtaining the tight analysis of the canonical tester for the \rhoIndepSet property.

Once we have \cref{lem:BSindepset} in hand, the analysis of the canonical tester for \rhoIndepSet is straightforward. 
Namely, with a simple union bound argument, we can show that when $S \subseteq V$ is a set of $|S| = s$ vertices drawn uniformly at random and $G$ is $\epsilon$-far from the \rhoIndepSet property, there is only a small probability that there is a fingerprint $F \in \mathcal{F}$ and associated container $C(F)$ where $F \subseteq S$ and $S$ contains at least $\rho s$ vertices from $C(F)$. 
This implies a correspondingly low bound on the probability that $S$ contains an independent set of size $\rho s$ and incorrectly accepts.

\paragraph{A hypergraph container lemma for satisfiability.}
To prove \cref{thm:sat}, we again want to prove that the natural canonical tester has high success probability even when it has small sample complexity.
In the case of satisfiability, however, what is sampled by the canonical tester is a set of variables of the input CSP, and it accepts if and only if the CSP restricted to the sampled variables is satisfiable.

The problem of testing satisfiability can be expressed as a problem on hypergraphs.
For any CSP $\phi$, we construct its associated hypergraph $H_\phi$ in the following standard way.
The hypergraph $H_\phi$ has $nk$ vertices.
We associate each vertex with a possible assignment to a variable of $\phi$. 
(So each variable of $\phi$ is associated with $k$ vertices in $H_\phi$.)
We then add a $q$-edge to the hypergraph $H_\phi$ for each variable assignment that falsifies a constraint in $\phi$.
See \Cref{fig:contruct-hypergraph-example} for an example of an edge coming from the constraint that evaluates to false only when the variables $x_1,x_2,x_3$ are all equal.
With this construction, a satisfying assignment to $\phi$ corresponds to an independent set on $n$ vertices in $H_\phi$ where each vertex in this independent set corresponds to a unique variable.
And when $\phi$ is $\epsilon$-far from satisfiable, every induced subhypergraph of $H_{\phi}$ that corresponds to an assignment to all $n$ variables contains at least $\epsilon \binom{n}{q}$ edges.

\begin{figure}[t]
    \centering
    \includegraphics[width=0.5\textwidth]{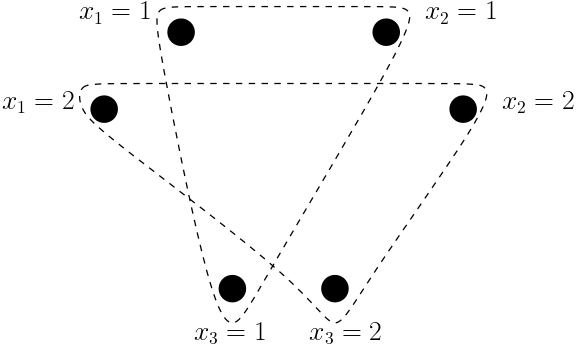}
    \caption{An illustration of the two hyperedges formed from the constraint that the variables $x_1$, $x_2$, and $x_3$ must not all be equal.}
    \label{fig:contruct-hypergraph-example}
\end{figure}

With this construction, the problem of analyzing the natural canonical tester for satisfiability is reduced to the problem of bounding the probability that a random set of $s$ variables corresponds to an independent set in $H_\phi$ with one vertex associated with each of these variables.
We introduce a hypergraph container lemma that lets us bound the latter probability.

\begin{lemma}
\label{lem:sat}
Fix any instance $\phi$ of \qkSAT on $n$ variables that is $\epsilon$-far from satisfiable.
Let $H_\phi$ be the hypergraph corresponding to $\phi$ obtained by the construction above.
Then there is a collection of fingerprints $\mathcal{F} \subseteq P(V[H_\phi])$ and a function $C: \mathcal{F} \rightarrow P(V[H_\phi]),$ which we call the container function, such that for every independent set $I$ in $H_\phi$ whose vertices all correspond to different variables in $\phi$, there is a $F \in \mathcal{F}$ that satisfies $F \subseteq I, I \setminus F \subseteq C(F)$, $|F| \le 8 \frac{kq}{\epsilon}$, and the number ${\rm vars}(C(F))$ of distinct variables corresponding to vertices in $C(F)$ is bounded by
\[
{\rm vars}(C(F)) \le \left( 1 - |F| \cdot \frac{\epsilon}{4 kq^2 \ln \frac{kq}{\epsilon}} \right) n.
\]
\end{lemma}

As was the case with \cref{lem:BSindepset}, there is a trade-off between the size of fingerprints and their corresponding containers which is critical for obtaining the strong bounds on the sample complexity for the canonical satisfiability tester.
However, one key difference with \Cref{lem:sat} is that the desired measure of the container is not the size $|C(F)|$ but rather the number of distinct variables corresponding to vertices in $C(F),$ which we denote by ${\rm vars}(C(F)).$

There are two challenges in establishing \cref{lem:sat}.
The first is that it applies to hypergraphs, not graphs.
This is a significant departure from the work in~\cite{blaisSeth}, as \cref{lem:BSindepset} was established with an iterative procedure that was specifically tailored to graphs.
Specifically, it relied heavily on the fact that for any vertex $v$ in a fingerprint, all the neighbors of $v$ can be eliminated from the corresponding container since an independent set that includes $v$ does not include any of its neighbours.
This simple elimination argument no longer holds even for 3-uniform hypergraphs, so a hypergraph container lemma needs to be obtained from a different process.

Second, in order to obtain the strongest bounds possible on the analysis of the canonical tester, we must also use the fact that the hypergraphs constructed from CSPs are not arbitrary, but rather have a specific structure.
In particular, if $\phi$ is $\epsilon$-far from satisfiable, then all $n$-subgraphs in $H_\phi$ that correspond to assignments to $n$ different variables have at least $\epsilon \binom{n}{q}$ edges. 
This specific characteristic of $H_\phi,$ which itself is very different than any characteristics utilized by hypergraph container methods in the past, can be exploited when constructing the container/fingerprint pairs.

\paragraph{A graph container lemma for independent set stars.}
To prove \cref{thm:cliques-queries}, we must first design a non-canonical algorithm for testing the \rhoIndepSet property.

The tester that we introduce is a natural variant of the canonical tester.
As with the canonical tester, it begins by drawing a set $S$ of vertices of the graph uniformly at random.
In the next step, it then draws a smaller subset $R \subseteq S$ that we will call the \emph{core} of the sample.
The algorithm then queries all pairs of vertices in $R \times S$. 
In other words, instead of querying all pairs of vertices in $S$ (as the canonical tester does), this algorithm queries only the pairs of vertices in $S$ that contain at least one vertex in the core $R$.
The algorithm then accepts if and only if there is an independent set $I \subseteq R$ of size $|I| = \rho |R|$ and a set $J \subseteq S$ of size $|J| = \rho |S|$ such that the vertices in $I$ are not connected to any of the vertices in $J$.

We can describe the tester in an equivalent alternative way with a little bit more notation.
A \emph{star} in a graph is a vertex $v$ along with a set $J \subseteq V \setminus \{v\}$ such that $v$ is connected to all the vertices in $J$.
Similarly, a \emph{star complement} is a vertex $v$ and a set $J \subseteq V \setminus \{v\}$ such that $v$ is not connected to any vertex in $J$.
An \emph{independent set star} is a natural generalization of this notion: it is a set $I \subseteq V$ that we call the \emph{core} along with a set $J \subseteq V \setminus I$ such that $I$ forms an independent set and none of the vertices in $J$ are connected to any vertex in $I$.
(See \cref{fig:indep-set-star} for an illustration of independent set stars.)
Then the tester we just defined accepts if and only if it finds an independent set star with a core of size $\rho|R|$ in $R$ and total size $\rho |S|$ in $S$.

\begin{figure}[t]
    \centering
    \includegraphics[width=0.4\textwidth]{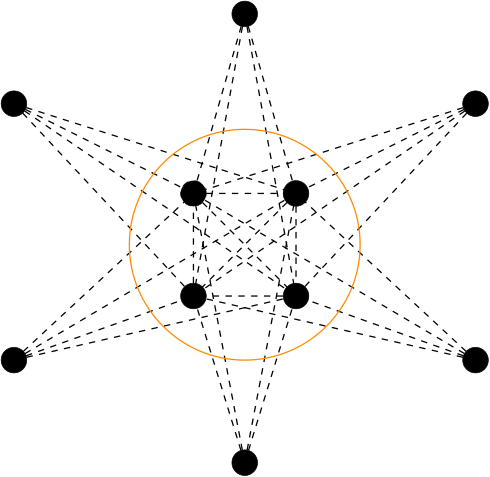}
    \caption{An independent set star with a core of size $4$ (circled) and $6$ outer vertices. Dashed lines represent pairs of vertices that are not adjacent in the graph. The outer vertices may or may not be adjacent to each other.}
    \label{fig:indep-set-star}
\end{figure}

The challenge in the analysis of the new tester is in showing that it accepts any graph that is $\epsilon$-far from \rhoIndepSet with only small probability.
\Cref{lem:BSindepset} is not of use here because it can only be used to bound the probability that a sample of vertices includes an independent set.
What we want now is a different tool to enable us to bound the probability that a different type of structure---namely, an independent set star---is present in the sampled subgraph.
This is what we obtain with the following graph container lemma for independent set stars.

\begin{lemma}
\label{lem:ISstar}
Let $G = (V,E)$ be a graph on $n$ vertices that is $\epsilon$-far from the \rhoIndepSet property.
Then there is a collection of fingerprints $\mathcal{F} \subseteq P(V)$ and two functions $C,D:\mathcal{F} \rightarrow P(V)$ such that for every independent set star $(I,J)$ with core $I$ in $G$, there exists a $F \in \mathcal{C}$ that satisfies $F \subseteq I \subseteq C(F) \subseteq D(F)$, $I\cup J \subseteq D(F)$, $|F| \le \frac{\rho^2}{\epsilon} \cdot 8\ln(\frac{2\rho}{\epsilon})$, 
\[
	|C(F)| \le \left( \rho - |F| \cdot \frac{\epsilon}{8\rho \ln \frac{2\rho}{\epsilon}} \right) n
\]
and if $|F| < 4\rho\ln(\frac{2\rho}{\epsilon})/\sqrt{\epsilon}$ then
\[
	|D(F)| \le \left( \rho - |F| \cdot \frac{\epsilon}{8\rho \ln \frac{2\rho}{\epsilon}} \right) n
\]
as well.
\end{lemma}

Informally, \cref{lem:ISstar} says that there is a small collection $\mathcal{C}=\{(C(F),D(F)) : F \in \mathcal{F}\}$ of container pairs (which we can think of as the inner and the outer containers) such that every independent set star $(I,J)$ in the graph $G$ is contained by some pair $(C,D) \in \mathcal{C}$.
Here, containment implies that the core $I$ is contained in the inner container $C$, and the entire independent set star $I \cup J$ is in the outer container $D$.
We also have the guarantee that the inner container is always small (as before with a trade-off between the size of a container and of its fingerprint) 
and the additional guarantee that for all the container pairs with small fingerprints, the outer container is also small.

As far as we know, \cref{lem:ISstar} is the first graph container lemma that applies directly to graph structures that are not independent sets. 
That is, in the prior applications of the graph and hypergraph container method that we are aware of, even when the container method has been used to enumerate other types of structures or combinatorial objects (such as square-free graphs, sum-free sets etc.), it was done by reformulating the problem as a problem of counting independent sets in an appropriately defined graph or hypergraph.

\section{Preliminaries}
\label{sect:prelim}

\paragraph{Testing graph properties.}
For any property $\Pi$ of graphs and any parameter $\epsilon > 0$, we say that the graph $G = (V,E)$ on $|V| = n$ vertices is \emph{$\epsilon$-far} from having property $\Pi$ if at least $\epsilon n^2$ edges need to be added or removed from the graph to turn it into another graph $G'$ that does have property $\Pi$.

A randomized algorithm $T$ is an \emph{$\epsilon$-tester} for the property $\Pi$ of graphs if there are absolute constants $c_0,c_1$ satisfying $0 \le c_0 < c_1 \le 1$ such that $T$ accepts all graphs that have property $\Pi$ with probability at least $c_1$ and accepts all graphs that are $\epsilon$-far from $\Pi$ with probability at most $c_0$.
(The definition of testers often fixes $c_0 = \frac13$ and $c_1 = \frac23$ but standard success amplification methods show that the two definitions result in complexity measures that are asymptotically equivalent so we use the more general definition for convenience.)

We are interested in property testers that examine only a sublinear fraction of the graph during their execution.
The most basic such tester is the \emph{canonical tester} that samples a set $S \subseteq V$ of vertices uniformly at random and accepts if and only if the induced subgraph $G[S]$ on the set $S$ of vertices satisfies some fixed predicate.
For a hereditary graph property $\Pi$, the \emph{natural canonical tester} is the one that accepts if and only if the induced subgraph has the property $\Pi$.
The \emph{sample complexity} of the canonical tester is the size $|S|$ of the set of sampled vertices.

A more general class of graph property testers is the class of algorithms that can query arbitrary pairs of vertices $v,w \in V$ to determine if $(v,w) \in E$ or not.
(This definition along with the choice of distance measure corresponds to the \emph{dense graph} model of property testing.)
A tester that fixes all of the edge queries it wishes to make to the unknown graph $G$ before obtaining the answer to any of these queries is called \emph{non-adaptive}.
An \emph{adaptive} tester can, by contrast, determine which edge to query next based on what the previous edge queries have revealed about the graph $G$.
The \emph{query complexity} of a non-adaptive or adaptive graph property tester is the maximum number of edge queries it may make during its execution.

Property testers give two natural measures of complexity for graph properties:
\begin{itemize}
	\item[$\cS_\Pi$:] The \emph{sample complexity} of $\Pi$, denoted $\cS_\Pi(n,\epsilon)$, is the minimum sample complexity of a canonical $\epsilon$-tester for $\Pi$ on graphs over $n$ vertices.
	\item[$\cQ_\Pi$:] The \emph{query complexity} of $\Pi$, denoted $\cQ_\Pi(n,\epsilon)$, is the minimum query complexity of an $\epsilon$-tester for $\Pi$ on graphs over $n$ vertices.
\end{itemize}

For any graph property $\Pi,$ a canonical tester with sample complexity $s$ makes $\binom{s}{2}$ queries, and so $\cQ_\Pi(n,\epsilon) \leq \cS_\Pi(n,\epsilon)^2.$
Further, there is at most a quadratic gap between the query complexity of the canonical tester and the query complexity of the optimal tester~\cite{GoldreichTrevisan03}, which can be rephrased as $\cQ_\Pi(n,\epsilon) = \Omega(\cS_\Pi(n,\epsilon)).$

\paragraph{Testing hypergraph properties.}
For any $q \ge 2$, a \emph{$q$-uniform hypergraph} $H = (V,E)$ is the natural generalization of a graph where elements of $E$ are sets of $q$ distinct elements from $V$. 
(When $q = 2$, this definition corresponds exactly to that of a graph.)
All the definitions from the last section carry over to properties of hypergraphs, with one important difference: for any fixed property $\Pi$ of $q$-uniform hypergraphs, the hypergraph $H$ is \emph{$\epsilon$-far} from $\Pi$ if we have to add or remove $\epsilon \binom{n}{q}$ edges from the hypergraph to obtain a hypergraph $H'$ that has property $\Pi$.
This modification ensures that we have that $H$ is $\epsilon$-far from a property $\Pi$ if and only if an $\epsilon$ fraction of all possible edges in a $q$-uniform hypergraph must be modified to get a hypergraph with the property.

For any subset of the vertices $U \subseteq V,$ we write $H[U]$ to denote the $q$-uniform hypergraph with vertex set $U$ and with an edge set that contains every $q$-edge of $H$ whose endpoints are all contained in $U.$

\begin{remark}
Much of the prior work on testing hypergraphs uses a slightly different definition of distance.
Namely, a $q$-uniform hypergraph is often defined to be $\epsilon$-far from a property if at least $\epsilon n^q$ edges (instead of $\epsilon \binom{n}{q}$ as above) need to be added or removed to obtain a hypergraph with the property.
This is often slightly more convenient for analysis (and is in fact why we use this version of distance for graph testing), however the two notions are asymptotically equivalent only when $q$ is a fixed constant.
However, in the setting that particularly interests us here where we care about the asymptotic dependence on $q$ as well as on other parameters in the property definition, our notion of distance is strictly more restrictive and more appropriate to the study of (simple) hypergraphs where all hyperedges are sets of distinct vertices.
In particular, a sample complexity with polynomial dependence on $q$ under our notion of distance also implies a matching polynomial dependence on $q$ under the more relaxed notion of distance, but the converse is not true.

We also note that, in either notion, a tester with sample complexity that is polynomial in $q,k$ and $1/\epsilon$ still queries a number of hyperedges of $H$ that is exponential in $q.$
This is due to the fact that, for a sample $S$ of vertices, the induced subgraph $H[S]$ has up to $\binom{|S|}{q}$ potential hyperedges.
\end{remark}

\paragraph{Testing satisfiability.}
An instance of a \emph{constraint satisfaction problem} (CSP) on a set of $n$ variables that can each take one of $k$ possible values is a set of predicates, which we call \emph{constraints}, on these variables.
A CSP is \emph{$q$-uniform} when all its constraints are predicates on exactly $q$ variables.

The CSP $\phi$ is \emph{satisfiable} if and only if there is an assignment of values to the variables such that all of the constraints are satisfied.
The \qkSAT property is the set of all $q$-uniform CSPs on $n$ variables that each can take one of $k$ possible values.
A $q$-uniform CSP $\phi$ on $n$ variables is \emph{$\epsilon$-far} from \qkSAT (or $\epsilon$-far from satisfiable) if and only if at least $\epsilon \binom{n}{q}$ constraints need to be removed to turn it into a satisfiable CSP.

For any set $S$ of variables, we write $\phi[S]$ to denote the restriction of the CSP $\phi$ to the set $S$.
The restriction $\phi[S]$ consists of the set of constraints whose variables are all present in $S$.

\paragraph{Semi-homogeneous graph partition properties.}
Graph partition properties have played a central role in property testing ever since the original work of Goldreich, Goldwasser, and Ron~\cite{goldreichPropertyTesting1998}.
A graph property $\Pi$ is a \emph{$k$-partition} property if it can be described by a matrix of $2k^2$ parameters $\ell_{1,1},\ldots,\ell_{k,k},u_{1,1},\ldots,u_{k,k}$ where $0 \le \ell_{i,j} \le u_{i,j} \le 1$ for each $i,j \in [k]$ and where a graph $G$ has property $\Pi$ if and only if there is a partition of the vertices of $G$ into the sets $V_1,\ldots,V_k$ for which the number of edges $e_G(V_i,V_j)$ between vertices in $V_i$ and $V_j$ satisfies
\[
	\ell_{i,j} \le \frac{e_G(V_i,V_j)}{|V_i| \cdot |V_j|} \le u_{i,j} 
\]
for all $i \neq j \in [k]$, and the number of edges $e_G(V_i)$ in the induced subgraph $G[V_i]$ satisfies
\[
	\ell_{i,i} \le \frac{e_G(V_i)}{\binom{|V_i|}{2}} \le u_{i,i} 
\]
for all $i \in [k]$.
(Note that there are different classes of partition properties in which one can add other types of constraints, such as absolute bounds on the sizes of the parts and/or on the number of edges between different parts. 
In this paper, however, we focus our attention on the class of partition properties described above.)

An interesting subclass of partition properties is that of \emph{homogeneous partition properties}, where for every $i,j \in [k]$ the density bounds satisfy $\ell_{i,j} = u_{i,j} \in \{0,1\}$. 
That is, for every pair of parts $V_i$ and $V_j$ and any graph that satisfies a homogeneous property, all vertices in $V_i$ must be connected to every vertex in $V_j$ (if $\ell_{i,j} = u_{i,j} = 1$), or none of those edges must be present in $G$ (if $\ell_{i,j} = u_{i,j} = 0$).
Examples of homogeneous graph properties include the property of being an empty graph (which is a $1$-partition property) and the property of being a collection of $k$ isolated cliques.

The class of \emph{semi-homogeneous partition properties} is a superclass of the set of all homogeneous partition properties obtained by relaxing the constraint on the densities to require $\ell_{i,j}, u_{i,j} \in \{0,1\}$ but not that $\ell_{i,j} = u_{i,j}$.
With this definition, the required density between some parts can be 0 or 1, as with homogeneous properties, but for some other pairs of parts we can have no restrictions.
The class of semi-homogeneous partition properties is large and includes many natural properties that are not homogeneous, including most notably $k$-colorability.

\section{Satisfiability Testing}
\label{sect:testing-sat}
In this section we prove \Cref{thm:sat}.
In particular, we are given a CSP $\phi,$ with $n$ variables, exactly $q$ variables per constraint and an alphabet of size $k,$ and wish to distinguish between the case that $\phi$ is satisfiable and the case that $\phi$ is $\epsilon$-far from satisfiable.
As in \cite{Sohler12}, we give a canonical tester that takes a sample $S$ of variables selected uniformly at random, and checks whether there exists an assignment to the variables in $S$ that satisfies $\phi[S].$
If $\phi$ is satisfiable, then $\phi[S]$ will certainly be satisfiable.
The main challenge of the proof is demonstrating that when $\phi$ is $\epsilon$-far from satisfiable, then $\phi[S]$ is satisfiable with low probability.

In order to handle this case, we start by describing a hypergraph container procedure and proving some basic properties.
Then, we define how to construct the hypergraph $H_\phi$ given a CSP $\phi,$ and prove a hypergraph container lemma for independent sets in $H_\phi$ when $\phi$ is $\epsilon$-far from satisfiable.
Finally, we use this hypergraph container lemma to prove \Cref{thm:sat}.

We note that throughout this section we assume that $q=o(n).$
We can do this because the bound on the sample complexity in \Cref{thm:sat} is trivial when $q=\Omega(n).$

\subsection{Container Procedure}
We now describe our hypergraph container procedure.
The container procedure takes as input a $q$-uniform hypergraph $H=(V,E),$ a parameter $n < |V|,$ and an independent set $I \subseteq V$, and produces a sequence of fingerprints and containers.
The procedure constructs this sequence iteratively.
After setting the initial container $C_0$ to be $V$ and the initial fingerprint $F_0$ to be empty, the procedure runs for $|I|/(q-1)$ iterations,\footnote{Note that for clarity of presentation, here and throughout we ignore all integer rounding issues as they do not affect the asymptotics of the final result.} where in each iteration the procedure selects $(q-1)$ distinct fingerprint vertices from the independent set and reduces the size of the container accordingly.
These fingerprint vertices are selected in a greedy fashion to try and remove as many vertices from the container as possible.

The basic idea for selecting these fingerprint vertices and updating the containers is as follows.
Suppose the current container is $C,$ and suppose $v$ is the vertex contained in $I$ with highest degree in $H[C].$
First observe that all vertices with degree higher than $v$ in $H[C]$ are certainly not in the independent set, and so they are removed from the container.
Further, every $q$-edge in $H[C]$ can be viewed as a set of $q$ vertices that cannot all be contained in the independent set.
By constructing a $(q-1)$-uniform hypergraph with vertex set $C \setminus v$ and with an edge set that is all formed by removing $v$ from all $q$-edges containing $v,$ we have that each $(q-1)$-edge in this new hypergraph can be viewed as $q-1$ vertices that cannot all be in the independent set.
And so we recurse on the new $(q-1)$-uniform hypergraph.

In order to optimize the procedure for hypergraphs constructed from CSPs, we modify this basic approach as follows.
Since our procedure will run on a hypergraph $H_\phi$ constructed from a CSP $\phi$ on $n$ variables that is $\epsilon$-far from satisfiable, and this hypergraph has the important property that every $n$-subgraph corresponding to unique variables is dense (see \Cref{sect:sat-GCL} for specific details), it is useful to consider the degree of a vertex in induced subgraphs on at most $n$ vertices.

In more detail, the first fingerprint vertex selected in each iteration will be the vertex which achieves the highest degree in some induced subgraph on at most $n$ vertices, which we call a $(\leq n)$-subgraph, of the current container.
For example, if the current container is $C \subseteq V[H],$ then for each vertex $v$ we calculate the quantity \[\deg_{H[C]}^{\leqslant n}(v) \coloneqq \max_{D \in \binom{C}{\leq n}}\left(\deg_{H[D]}(v)\right)\] and select the fingerprint vertex to be the vertex of $I$ achieving the highest value, where ties are broken based on an assumed ordering of the vertices $V[H].$

After selecting the fingerprint vertex $v_{t,q},$ where $t$ denotes that we are in the $t$-th iteration and the $q$ denotes the fact that the vertex was selected from a $q$-uniform hypergraph, the rest of the iteration consists of running an inner loop that begins with the subgraph that achieves this maximum.
In particular, construct a $(q-1)$-uniform hypergraph $\cD_{t,q-1}$ with vertex set \[ V[\cD_{t,q-1}] = \argmax_{D \in \binom{C}{\leq n}} \left(\deg_{H[D]}(v_{t,q})\right) \setminus v_{t,q},\] where if there are multiple $D$ sets reaching the maximum then an arbitrary maximal set is selected.
The edge set of $\cD_{t,q-1}$ is constructed by removing $v_{t,q}$ from all $q$-edges of $H[C]$ that contain $v_{t,q}.$
The edge set corresponds to sets of vertices that cannot all be part of the independent set.

The inner loop follows the basic approach described above.
In particular, the loop runs from $\ell=q-1, \dots, 2,$ and in each step it selects the vertex $v_{t,\ell}$ in $I$ with highest degree in the current hypergraph $\cD_{t,\ell}$ and adds it to the fingerprint.
It then constructs a new $(\ell-1)$-uniform hypergraph $\cD_{t,\ell-1}$ that has edge set consisting of all edges of $\cD_{t,\ell}$ containing $v_{t,\ell},$ with $v_{t,\ell}$ removed.
In each step, we also keep track of the vertices with higher degree than the vertex selected, and these vertices are removed after the inner loop finishes because they cannot be part of the independent set.

Finally, after the inner loop is completed, we are left with a $1$-uniform hypergraph $\cD_{t,1}.$
Since, in each step, the edge set corresponds to sets of vertices that cannot all be part of the independent set, we have that all the $1$-edges in $\cD_{t,1}$ correspond to vertices that cannot be part of the independent set. So these vertices are also removed.
This completes one iteration, and this is repeated $|I|/(q-1)$ times until all vertices of $I$ have been selected.

The procedure outputs the sequence of fingerprints $F_1,F_2,\ldots,F_{|I|/(q-1)}$ and their corresponding containers $C_1,C_2,\ldots,C_{|I|/(q-1)}$ generated by the algorithm for the independent set $I$.
We refer to $F_t$ and $C_t$ as the \emph{$t$-th fingerprint} and \emph{$t$-th container} of $I$, respectively.
When we consider the fingerprints or containers of multiple independent sets, we write $F_t(I)$ and $C_t(I)$ to denote the $t$-th fingerprint and container of the independent set $I$.
(When the current context specifies a single independent set, we will continue to write only $F_t$ and $C_t$ to keep the notation lighter.)
We note that it is also convenient to extend the definition with $C_t=F_t=I$ for all $t > |I|/(q-1).$

\bigskip
\begin{algorithm}[!htbp]
\caption{\sc{Fingerprint \& Container Generator}}
\label{alg:FingerprintContainer}
\smallskip
\KwIn{A $q$-uniform hypergraph $H = (V,E),$ a parameter $n < |V|,$ and an independent set $I \subseteq V$ }
\medskip
Initialize $F_0 \gets \emptyset$ and $C_0 \gets V(H)$\;

\For{$t=1,2,\dots,|I|/(q-1)$} {
    $v_{t,q} \gets$ the vertex in $I \cap C_{t-1}$ that has the largest value of $\deg_{H[C_{t-1}]}^{\leqslant n}(v)$\;
    
    \smallskip
    \tcp{Keep track of all vertices with higher $q$-degree than $v_{t,q}$ to be removed later}
    $X_q \gets \big\{ w \in C_{t-1} : \deg_{H[C_{t-1}]}^{\leqslant n}(w) > \deg_{H[C_{t-1}]}^{\leqslant n}(v_{t,q}) \big\}$\;
    
    \smallskip
    \tcp{Construct $(q-1)$-uniform hypergraph $\cD_{t,q-1}$ on the set vertices of $C_{t-1},$ of size at most $n,$ that $v_{t,q}$ has highest degree in. The edge set of $\cD_{t,q-1}$ is all edges incident to $v_{t,q},$ with $v_{t,q}$ removed}
    $V[\cD_{t,q-1}] \gets \left(\argmax_{D \in \binom{C_{t-1}}{\leq n}}\deg_{H[D]}(v_{t,q})\right) \setminus v_{t,q}$\;
    $E[\cD_{t,q-1}] \gets \big\{ e \setminus v_{t,q}: e \in E[H[C_{t-1}]] \textrm{ and } v_{t,q} \in e \big\}$\;
    
    \smallskip
    \For{$\ell = q-1, \dots, 2$} {
        $v_{t,\ell} \gets$ the vertex in $I \cap V[\cD_{t,\ell}]$ with largest $\ell$-degree in $\cD_{t,\ell}$;

        \smallskip
        \tcp{Keep track of all vertices with higher $\ell$-degree than $v_{t,\ell}$ in $\cD_t$ to remove later}
        $X_\ell \gets \big\{ w \in V[\cD_{t,\ell}] : \deg_{\cD_{t,\ell}}(w) > \deg_{\cD_{t,\ell}}(v_{t,\ell}) \big\}$\;
        
        \smallskip
        \tcp{Construct $\cD_{t,\ell-1}$ on the vertices of $\cD_{t,\ell}$ minus $v_{t,\ell},$ and containing all edges incident to $v_{t,\ell},$ with $v_{t,\ell}$ removed}
        $V[\cD_{t,\ell-1}] \gets V[\cD_{t,\ell}] \setminus v_{t,\ell}$\;
        $E[\cD_{t,\ell-1}] \gets \big\{ e \setminus v_{t,\ell}: e \in E[\cD_{t,\ell}] \textrm{ and } v_{t,\ell} \in e\big\}$\;
    }
    \smallskip 
    \tcp{Add $v_{t,2},\dots,v_{t,q}$ to the fingerprint}
    $F_t \gets F_{t-1} \cup \{v_{t,2},\dots,v_{t,q}\}$\;
    
    \smallskip
    \tcp{Update container set by removing all $X_\ell$ sets, all vertices with a $1$-edge, and the new fingerprint vertices}
    $C_t \gets C_{t-1} \setminus \left(\cup_{\ell =2}^q X_\ell \cup \{v \in C_{t-1}: \{v\} \in E[\cD_{t,1}]\} \cup \{v_{t,2},\dots,v_{t,q}\}\right)$\;
    
}
Return $F_1,\ldots,F_{|I|/q}$ and $C_1,\ldots,C_{|I|/q}$\;
\end{algorithm}
\bigskip

\subsection{Basic Properties}
By construction, the sequence of fingerprints and containers of an independent set $I$ satisfy a number of useful properties.
For instance, we have $F_1 \subseteq F_2 \subseteq \cdots \subseteq F_{|I|/(q-1)} \subseteq I$ and $I \setminus F_{|I|/(q-1)} \subseteq C_{|I|/(q-1)} \subseteq C_{|I|/(q-1)-1} \subseteq \cdots \subseteq C_2 \subseteq C_1,$
so that for each $t$, $F_t \subseteq I$ and $I \setminus F_t \subseteq C_t$.

The construction of the algorithm also guarantees that the $t$-th container of an independent set is the same as the $t$-th container of its $t$-th fingerprint.
In other words, even for a large independent set, the $t$-th container is uniquely defined by the $t$-th fingerprint of size $t \cdot (q-1),$ which may be much smaller than the size of the independent set.
This can be used to show an upper bound on the number of possible containers, which will be useful later in the proof of \Cref{thm:sat}.

\begin{proposition}
\label{prop:container-closure}
For any hypergraph $H = (V,E)$, any independent set $I$ in $H$, and any $t,$ the fingerprint $F_t(I)$ and container $C_t(I)$ of $I$ satisfy 
\[
C_t\big( F_t(I) \big) = C_t(I).
\]
\end{proposition}

\begin{proof}
If $t> |I|/(q-1)$ then $F_t(I)=I=C_t(I)$ and the equality holds.
Now consider the case that $t \leq |I|/(q-1).$ 
If $v_{1,q},\ldots v_{1,2},v_{2,q}, \ldots,v_{2,2},\ldots,v_{t,q},\ldots,v_{t,2}$ are the vertices selected in the first $t$ iterations of the container algorithm on input $I$, then the algorithm selects the same vertices and forms the same container $C_t$ when provided with input $F_t(I)$ instead of $I$.
\end{proof}

Next, we demonstrate that the container procedure makes progress in every step of the procedure by finding sparser and sparser containers.
In order to show this, we first require a basic property about the number of vertices with nearly average degree in a hypergraph.
The average degree in a $\ell$-uniform hypergraph on $|E|$ edges is $\ell|E|/|V|.$
In the following proposition, we show that if $|E|$ is sufficiently large, then a large number of vertices have degree at least $(\ell-1)|E|/|V|.$

\begin{proposition}
\label{prop:edges-bound}
    Let $H=(V,E)$ be a $\ell$-uniform hypergraph, with $\ell \geq 2.$
    Then there are at least $|E|/\binom{|V|-1}{\ell-1}$ vertices in $H$ with degree larger than $(\ell-1)|E|/|V|.$
\end{proposition}
\begin{proof}
Suppose not.
In other words, suppose that there are strictly less than $|E|/\binom{|V|-1}{\ell-1}$ vertices with degree larger than $(\ell-1)|E|/|V|.$
These vertices have degree at most $\binom{|V|-1}{\ell-1}.$
The remaining vertices, of which there are at most $|V|,$ have degree at most $(\ell-1)|E|/|V|.$
So the total number of edges is less than $\frac{1}{\ell} \left(\frac{|E|}{\binom{|V|-1}{\ell-1}} \cdot \binom{|V|-1}{\ell-1} + |V| \cdot \frac{(\ell-1)|E|}{|V|}\right) = |E|,$ a contradiction.    
\end{proof}

We can now state and prove the property that demonstrates that the container procedure makes progress in every iteration.
In particular, we show that as $t$ increases, we get a stronger upper bound on the maximum degree of every $(\leq n)$-subgraph of $H[C_t].$

\begin{proposition}
\label{prop:container-degree}
Let $H=(V,E)$ be a $q$-uniform hypergraph with $|V|=kn,$ let $I$ be an independent set in $H,$ and let $t \leq |I|/(q-1).$
Then the maximum $q$-degree of every $(\leq n)$-subgraph of $H[C_t]$ is at most $\frac{2kq}{t}\binom{n-1}{q-1}.$
\end{proposition}

\begin{proof}
Fix any $t \le |I|/(q-1)$.
Since the containers satisfy the containment $V = C_0 \supseteq C_1 \supseteq \dots \supseteq C_t$, then
\[
kn \ge |C_0| - |C_t| = \sum_{i=1}^t \big( |C_{i-1}| - |C_i| \big)= 
\sum_{i=1}^t |C_{i-1} \setminus C_i|
\]
and there must be an index $j \in [t]$ for which $|C_{j-1} \setminus C_j| \leq kn/t$.

We claim that for each step of the inner loop of the $j$th iteration, the number of $\ell$-edges in $\cD_{j,\ell}$ is at most $\frac{2k\ell}{t} \binom{n-(q-\ell)}{\ell}.$
First we show that this claim suffices to prove the proposition.
Using the claim, the number of $(q-1)$-edges in $\cD_{j,q-1}$ is at most $\frac{2k(q-1)}{t} \binom{n-1}{q-1},$ and so the fingerprint vertex $v_{j,q}$ has $q$-degree at most $\frac{2k(q-1)}{t} \binom{n-1}{q-1}$ in every induced $(\leq n)$-subgraph of $H[C_{j-1}].$
Hence, all vertices that had higher degree in some $(\leq n)$-subgraph of $H[C_{j-1}]$ were excluded from $C_j$, so the maximum degree in any $(\leq n)$-subgraph of $H[C_j]$ is at most $\frac{2k(q-1)}{t} \binom{n-1}{q-1}.$
Since $C_t \subseteq C_j$, the maximum degree of any $(\leq n)$-subgraph of $H[C_t]$ is also bounded above by $\frac{2k(q-1)}{t} \binom{n-1}{q-1}$ as well.

We now proceed to prove the claim by induction. In the $\ell=1$ base case, the vertices corresponding to $1$-edges in $\cD_{j,1}$ are removed from $C_j.$
Since $|C_{j-1} \setminus C_j| \leq \frac{kn}{t}$ then the number of $1$-edges is at most $\frac{kn}{t},$ which is at most $\frac{2k (n-q+1)}{t}$ since $q=o(n).$
In the inductive case $\ell \leq q-1,$ suppose that the number of $\ell$-edges in $\cD_{j,\ell}$ is more than $\frac{2k\ell}{t} \binom{n-(q-\ell)}{\ell}.$
If $v_{j,\ell}$ has degree larger than $\frac{2k(\ell-1)}{t} \binom{n-(q-\ell)-1}{\ell-1},$ then, by the construction of $\cD_{j,\ell-1},$ the number of $(\ell-1)$-edges in $\cD_{j,\ell-1}$ is larger than $\frac{2k(\ell-1)}{t} \binom{n-(q-\ell)-1}{\ell-1},$ a contradiction with the inductive hypothesis.
Otherwise, by \Cref{prop:edges-bound}, there are at least \[\frac{|E[\cD_{j,\ell}]|}{\binom{|V[\cD_{j,\ell}]|-1}{\ell-1}}\geq \frac{\frac{2k\ell}{t} \binom{n-(q-\ell)}{\ell}}{\binom{n-(q-\ell)-1}{\ell-1}} =\frac{2k(n-(q-\ell))}{t}>\frac{kn}{t}\] vertices with degree larger than $\frac{(\ell-1)|E[\cD_{j,\ell}]|}{|V[\cD_{j,\ell}]|} \geq \frac{2(\ell-1) k\ell \binom{n-(q-\ell)}{\ell}}{t (n-(q-\ell))}=\frac{2k(\ell-1)}{t}\binom{n-(q-\ell)-1}{\ell-1}$ in $\cD_{j,\ell},$ where we use that $|V[\cD_{j,\ell}]| \leq n-(q-\ell)$ in both inequalities and that $q=o(n).$
These high degree vertices are all removed from $C_j$ because they have higher degree than $v_{j,\ell}.$
This results in $|C_{j-1} \setminus C_j| > \frac{kn}{t},$ a contradiction, which completes the proof of the claim.
\end{proof}

\subsection{Hypergraph Container Lemma}
\label{sect:sat-GCL}
Our results on testing satisfiability make use of a standard transformation of CSPs to hypergraphs that we define here.

\begin{definition}
Let $\phi$ be a $q$-uniform CSP on $n$ variables that each can take $k$ possible values.
The \emph{hypergraph representation of $\phi,$} denoted $H_\phi,$ is defined to be a $q$-uniform hypergraph on $k \cdot n$ vertices, where each vertex corresponds to a single variable assignment.
For each constraint of $\phi$, we add one edge for each assignment to the $q$ variables that falsifies the constraint.
\end{definition}

See \Cref{fig:contruct-hypergraph-example} for an example of an edge coming from the constraint that evaluates to false only when the variables $x_1,x_2,x_3$ are all equal.
It's worth noting that edges are not added to sets of vertices that correspond to assignments to repeated variables.
Further, without loss of generality we may assume that for any set of $q$ variables there is at most one constraint on those variables in $\phi,$ and so $H_\phi$ does not have multi-edges.
Throughout this section we utilize the fact that each vertex corresponds to a specific single variable assignment and each $q$-edge can be viewed as an assignment to $q$ variables.

Using this construction, observe that for any subset $U \subseteq V[H_\phi]$ that only contains vertices corresponding to distinct variables, the edges of $H_\phi[U]$ correspond to distinct falsified constraints.
In other words, we have the following property.

\begin{observation}
    \label{prop:edgeConstraintEquiv}
    Let $U \subseteq V$ be a subset of vertices of $H_\phi$ such that the labels of vertices in $U$ are to distinct variables and let $A$ be the $|U|$-variable assignment corresponding to the labels of $U.$ Then the number of edges in $H_\phi[U]$ is exactly the number of constraints, on variables of $U,$ that are falsified by $A.$
\end{observation}

In particular, we have that $\phi$ is satisfiable if and only if there exists an independent set $I$ in $H_\phi$ with $|I| = n$ and where the vertices of $I$ correspond to distinct variables.
Further, if $\phi$ is $\epsilon$-far from satisfiable, then every induced subgraph of $H_\phi$ that corresponds to an assignment to all $n$ variables contains at least $\epsilon \binom{n}{q}$ edges.

Using this construction, we prove the following hypergraph container lemma.

\begin{lemma}[Reformulation of \Cref{lem:sat} \protect\footnote{The collection of fingerprints/containers in \Cref{lem:sat} is constructed by applying \Cref{lem:GCL-sat} to every independent set in the graph. The properties of the containers follow from the basic properties of the container procedure and \Cref{prop:container-closure}.}]
\label{lem:GCL-sat}
Let $\phi$ be an instance of \qkSAT on $n$ variables.
Let $H_\phi$ be the hypergraph representation of $\phi,$ and let $I$ be an independent set in $H_\phi.$
If $\phi$ is $\epsilon$-far from satisfiable then there exists $t \leq 8kq/\epsilon$ such that the number of distinct variables appearing in labels of $C_t(I),$ denoted ${\rm vars}(C_t(I)),$ is bounded by
\[ {\rm vars}(C_t(I)) \leq \left(1-\frac{\epsilon t}{4kq^2 \ln(kq/\epsilon)}\right)n.\]
\end{lemma}

\begin{proof}
    We prove the contrapositive.
    That is, assume that for all $t \le 8kq/\epsilon$, the number of distinct variables in labels of $C_t(I)$ is bounded by
    \begin{equation}
        \label{eqn:distinctVars}
        {\rm vars}(C_t(I)) > \left(1-\frac{\epsilon t}{4kq^2 \ln(kq/\epsilon)}\right)n.
    \end{equation}
    We will construct an assignment to the $n$ variables of $\phi$ which shows that $\phi$ is not $\epsilon$-far from satisfiable.
    
    Start with an assignment to as many distinct variables as possible coming from the labels of $C_{8kq/\epsilon}(I).$
    In other words, let $Z$ be the set of variables appearing in labels of $C_{8kq/\epsilon}(I),$ and construct $U \subseteq C_{8kq/\epsilon}(I)$ by removing vertices from $C_{8kq/\epsilon}(I)$ until each variable in $Z$ appears in exactly one label in $U.$

    For each $t=1,\dots,8kq/\epsilon,$ let $Z_t$ be the set of variables that appear in at least one label of a vertex in $C_{t-1},$ but do not appear in any labels of vertices in $C_t.$
    Further, for each $t=1,\dots,8kq/\epsilon,$ construct $V_t \subset V$  as follows.
    For each variable $z \in Z_t,$ add a single vertex from $C_{t-1}$ to $V_t$ that is labelled with an assignment of $z.$
    Note that if a variable in $Z_t$ appears in the labels of multiple vertices $C_{t-1},$ then we select one of the vertices arbitrarily and add it to $V_t.$

    First observe that $Z \cup Z_1 \cup \dots \cup Z_{8kq/\epsilon}$ form a partitioning of the $n$ variables of $\phi.$
    Hence, $U \cup V_1 \cup \dots \cup V_{8kq/\epsilon}$ contains $n$ distinct vertices, with the labels corresponding to a complete assignment $A$ to the $n$ variables of $\phi.$

    We claim that $H_\phi[U \cup V_1 \cup \dots \cup V_{8kq/\epsilon}]$ contains less than $\epsilon \binom{n}{q}$ edges.
    Before proving the claim, observe that this is sufficient for completing the proof of the lemma because, by \Cref{prop:edgeConstraintEquiv}, the number of constraints falsified by $A$ is less than $\epsilon \binom{n}{q}.$
    Since $A$ is an assignment to the $n$ variables of $\phi,$ then $\phi$ is not $\epsilon$-far from satisfiable.

    We now proceed to prove the claim.
    Since $U \subseteq C_{8kq/\epsilon}(I)$ and $U$ is of size at most $n,$ then, by \Cref{prop:container-degree}, the maximum degree in the subgraph $H[U]$ is at most $\frac{2kq \binom{n-1}{q-1}}{8kq/\epsilon}=\epsilon \binom{n-1}{q-1}/4,$ and so $H[U]$ has at most $\frac{\epsilon \binom{n}{q}}{4}$ $q$-edges.
    In order to count the additional edges in $H_\phi[U \cup V_1 \cup \dots \cup V_{8kq/\epsilon}],$ consider adding the vertices from $V_1,\dots,V_{8kq/\epsilon}$ to $U$ in reverse order (starting at vertices in $V_{8kq/\epsilon}$ and going to $V_1$).
    In this way, for any $t \leq 8kq/\epsilon,$ when $V_t$ is added, it holds that $\left(U \cup V_{8kq/\epsilon} \cup \dots \cup V_t\right) \subseteq C_{t-1}(I).$
    Hence, by \Cref{prop:container-degree}, each $v \in V_t$ contributes at most $\frac{2kq\binom{n-1}{q-1}}{t-1}$ additional edges (or $\binom{n-1}{q-1}$ edges if $t=1$) because $v$ is contained in $C_{t-1}(I).$
    
    Then, the sum of the edges in $H_\phi[U \cup V_1 \cup \dots \cup V_{8kq/\epsilon}]$ can be upper bounded by
    \begin{equation}
        \label{eqn:sum-edges}
        \frac{\epsilon \binom{n}{q}}{4}+|V_1|\cdot \binom{n-1}{q-1}+\sum_{t = 2}^{8kq/\epsilon} |V_t|\frac{2kq\binom{n-1}{q-1}}{t-1}.
    \end{equation}
    
    For any $t \leq 8kq/\epsilon,$ $\bigcup_{j = 1}^t Z_j$ is the set of variables not appearing in labels of $C_t(I),$ and so \eqref{eqn:distinctVars} implies that $\sum_{j= 1}^{t} |Z_j| \leq \frac{t \epsilon n}{4kq^2\ln(kq/\epsilon)}.$
    Since $|Z_j|=|V_j|$ for all $j,$ then $\sum_{j = 1}^{t} |V_j| \leq \frac{t \epsilon n}{4kq^2\ln(kq/\epsilon)}.$
    In the sum \eqref{eqn:sum-edges}, the contribution from $|V_t|$ goes down as $t$ increases, and so the sum is maximized when $|V_t|=\frac{\epsilon n}{4kq^2\ln(kq/\epsilon)}$ for all $t.$
    Hence, the sum of the edges in $H_\phi[U \cup V_1 \cup \dots \cup V_{8kq/\epsilon}]$ can be upper bounded by
    \[\frac{\epsilon \binom{n}{q}}{4}+\frac{\epsilon \binom{n}{q}}{4 kq\ln(kq/\epsilon)} + \frac{\epsilon \binom{n}{q}}{ 2\ln(kq/\epsilon)}\sum_{t=2}^{8kq/\epsilon} \frac{1}{t-1} < \epsilon \binom{n}{q},\]
    where the last inequality uses the upper bound $H(m) \le \ln(m) + 1$ on the harmonic series and the bound $\epsilon < e^{-8}$ that we can apply without loss of generality.
\end{proof}

\subsection{Proof of \autoref{thm:sat}}
\label{sect:sat-proof}
We can now complete the proof of \Cref{thm:sat}, restated below.

\newtheorem*{sat-thm-precise}{Theorem 1}
\begin{sat-thm-precise}[Precise formulation]
The sample complexity of $\epsilon$-testing the \qkSAT property is 
\[
\mathcal{S}_{\qkSAT}(n,\epsilon) = O\left(\frac{k q^3}{\epsilon} \ln^2(kq/\epsilon)\right).
\]
\end{sat-thm-precise}

\begin{proof}
    Let $S$ be a random set of $s = c\frac{k q}{\epsilon} \ln^2(kq/\epsilon)$ variables drawn uniformly at random without replacement from the variable set of $\phi,$ where $c$ is a large enough constant.
    The tester accepts a \qkSAT instance $\phi$ if and only if $\phi[S]$ is satisfiable, where $\phi[S]$ represents the constraints of $\phi$ restricted to $S.$ When $\phi$ is satisfiable, then $\phi[S]$ is also satisfiable and so the tester always accepts.

    In the remainder of the proof, we upper bound the probability that $\phi[S]$ is satisfiable when $\phi$ is $\epsilon$-far from satisfiable. Let us denote the variables of the sampled set $S$ as $u_1,u_2,\ldots,u_s.$ In the following, let us say that a container $C_t(I)$ is \emph{small} when the number of distinct variables in the labels of $C_t(I)$ is at most $\left(1-\frac{\epsilon t}{4kq^2\ln(kq/\epsilon)}\right)n,$ the expression in the conclusion of \Cref{lem:GCL-sat}.
    
    First observe that if $\phi[S]$ is satisfiable, then there exists an assignment to the variables $S$ that corresponds to an independent set $I$ in $H_\phi.$
    By \Cref{lem:GCL-sat}, for \emph{every} independent set $I$ in $H_\phi$, there is a $t \le 8kq/\epsilon$ and a fingerprint $F_t \subseteq I$ of size $t(q-1)$ that defines a small container $C_t(F_t)$ such that $I\setminus F_t \subseteq C_t(F_t).$
    Hence, the probability that $\phi[S]$ is satisfiable is at most the probability that there exists some $t \le 8kq/\epsilon$ and an assignment to $t(q-1)$ variables of $S$ which forms the fingerprint $F_t$ of a small container $C_t,$ and the remaining $s-t(q-1)$ variables of $S$ all comes from the labels of $C_t.$
    We now upper bound this probability. 

    For any $t \leq \frac{8kq}{\epsilon}$ and $t \cdot (q-1)$ distinct indices $i_1,i_2,\ldots,i_{t \cdot (q-1)} \in [s]$, let us consider the event where there exists an assignment to the variables $u_{i_1},\ldots,u_{i_{t(q-1)}}$ that corresponds to the labels of a fingerprint $F_t$ of a small container $C_t$. When this event occurs, the probability that the remaining $s - t(q-1)$ variables in $S$ are all drawn from the variables in the labels of $C_t(F_t)$ is at most $ \left( 1 - \frac{\epsilon t}{4kq^2\ln (kq/\epsilon)} \right)^{s-t(q-1)}$.

    So by applying a union bound argument over all $t \le \frac{8kq}{\epsilon},$ all possible choices of indices and all possible assignments to $t \cdot (q-1)$ variables, the probability that there exists an assignment to any set of at most $\frac{8kq(q-1)}{\epsilon}$ variables in $S$ that corresponds to the labelling of a fingerprint of a small container from which we sample all the remaining variables is at most 
    \[\sum_{t=1}^{8kq/\epsilon} \binom{s}{t (q-1)} \cdot k^{t (q-1)} \cdot \left(1-\frac{\epsilon t}{4kq^2\ln(kq/\epsilon)}\right)^{s-t(q-1)}
    \le \sum_{t=1}^{8kq/\epsilon}\exp\left(tq\ln(ks)-\frac{\epsilon t s}{8kq^2\ln(kq/\epsilon)}\right)
    \]
    where the inequality uses the fact that $s > 2tq$ and upper bound $\binom{s}{t (q-1)} \leq s^{tq}.$
    Using the fact that $s = c\frac{kq^3}{\epsilon} \ln^2(kq/\epsilon)$ for a large enough constant $c,$ the above expression can be upper bounded by $1/3.$
    Hence, the probability that $\phi[S]$ is satisfiable is at most $1/3.$
\end{proof}

\subsection{Proofs of \texorpdfstring{\Cref{cor:k-col,cor:01-partition}}{\autoref{cor:k-col} and \autoref{cor:01-partition}}}
\label{sect:sat-cor-proofs}
We complete this section by proving \Cref{cor:k-col,cor:01-partition} using \Cref{thm:sat}.
We prove \Cref{cor:k-col} by showing that testing whether a $q$-uniform hypergraph is $k$-colorable is a special case of testing $\qkSAT.$
Similarly, we prove \Cref{cor:01-partition} by showing that testing whether a graph $G$ satisfies a property in \kSHPP is a special case of testing \qtwokSAT.

\newtheorem*{sat-cor1-precise}{Corollary 2}
\begin{sat-cor1-precise}[Precise formulation]
The sample complexity of $\epsilon$-testing the \qkCol property is 
\[
\mathcal{S}_{\qkCol}(n,\epsilon) = O\left(\frac{kq^3}{\epsilon}\ln^2(kq/\epsilon)\right).
\]
\end{sat-cor1-precise}
\begin{proof}
    Let $H$ be any $q$-uniform hypergraph on $n$ vertices.
    We wish to distinguish between the case that $H$ is $k$-colorable and the case that $H$ is $\epsilon$-far from $k$-colorable.
    Let $\phi$ be a CSP on $n$ variables, one for each vertex, that can each take values from $[k].$
    The constraints of $\phi$ correspond to the hyperedges of $H.$
    In particular, for every hyperedge in $H$, there is a constraint to the variables associated with the vertices of the edge so that the variables cannot all be assigned the same value from $[k].$
    Using that edges of $H$ correspond exactly to constraints of $\phi$, then $\phi$ is satisfiable if and only if $H$ is $k$-colorable, and $\phi$ is $\epsilon$-far from satisfiable if and only if $H$ is $\epsilon$-far from $k$-colorable.

    We can test whether $H$ is $k$-colorable by simulating the tester from \Cref{thm:sat} on $\phi.$
    In particular, a sample of variables of $\phi$ corresponds exactly to a sample of vertices of $H,$ and when the tester asks to inspect $\phi[S],$ it can be simulated by inspecting $H[\hat{S}],$ where $\hat{S}$ is the vertices associated with the variables of $S.$ 
\end{proof}

\newtheorem*{sat-cor2-precise}{Corollary 3}
\begin{sat-cor2-precise}[Precise formulation]
The sample complexity of $\epsilon$-testing any property $\Pi$ in \kSHPP is 
\[
\mathcal{S}_{\Pi}(n,\epsilon) = O\left(\frac{k}{\epsilon}\ln^2(k/\epsilon)\right).
\]
\end{sat-cor2-precise}
\begin{proof}
    Let $\Pi$ be property in \kSHPP with parameters $\ell_{1,1},\ldots,\ell_{k,k},u_{1,1},\ldots,u_{k,k}$ where $\ell_{i,j}, u_{i,j} \in \{0,1\}$ for each $i,j \in [k].$
    Let $G=(V,E)$ be any graph on $n$ vertices.
    We wish to distinguish between the case that $G$ is in property $\Pi,$ and the case that $G$ is $\epsilon$-far from $\Pi.$
    
    Let $\phi$ be a CSP on $n$ variables $x_1,\dots,x_n$, one for each vertex $v_1,\dots,v_n \in V$, that can each take values from $[k].$
    In order to define the constraints of $\phi,$ we define the following two sets $\Pi_0,\Pi_1 \subseteq [k]^2.$
    Let $\Pi_0=\{(i,j) : i,j \in [k] \textrm { and } \ell_{i,j}=0\}$ and let $\Pi_1=\{(i,j) : i,j \in [k] \textrm { and } u_{i,j}=1\}.$
    In other words, $\Pi_0$ corresponds to possible partition assignments (from $[k]$) to pairs of vertices such that if the vertices are not neighbours then they are consistent with the desired edge densities of $\Pi.$
    Similarly, $\Pi_1$ corresponds to assignments (from $[k]$) to pairs of vertices such that if the vertices \emph{are} neighbours then they are consistent with the desired edge densities of $\Pi.$
    In this way, if it's possible to construct a partition of $V$ where every pair of non-adjacent vertices in $G$ is assigned to pairs of parts from $\Pi_0,$ and every pair of adjacent vertices in $G$ is assigned to pairs of parts from $\Pi_1,$ then $G$ is in \kSHPP.

    Using $\Pi_0$ and $\Pi_1,$ we define the constraints as follows.
    For every pair of non-adjacent vertices $(v_i,v_j)$ in $G,$ $\phi$ has a constraint on $x_i,x_j$ that evaluates to true if and only if $x_i$ and $x_j$ take pairs of assignments from $\Pi_0.$
    Similarly, for every pair of adjacent vertices $(v_i,v_j)$ in $G,$ $\phi$ has a constraint on $x_i,x_j$ that evaluates to true if and only if $x_i$ and $x_j$ take pairs of assignments from $\Pi_1.$

    First observe that $\phi$ is satisfiable if and only if $G$ is in \kSHPP.
    Further, if $G$ is $\epsilon$-far from satisfiable, then, for any partitioning of $G,$ at least $\epsilon n^2$ pairs of vertices violate the desired edge densities of $\Pi.$
    This corresponds to exactly $\epsilon n^2 > \epsilon \binom{n}{2}$ constraints of $\phi$ evaluating to false for any assignment to the $n$ variables, which means that $\phi$ is $\epsilon$-far from satisfiable. 

    We can test whether $G$ is in \kSHPP by simulating the tester from \Cref{thm:sat} on $\phi.$
    In particular, a sample of variables of $\phi$ corresponds exactly to a sample of vertices of $G,$ and when the tester asks to inspect $\phi[S],$ it can be simulated by inspecting $G[\hat{S}],$ where $\hat{S}$ is the vertices associated with the variables of $S.$
\end{proof}

\section{Query Complexity for Testing Independent Sets}
\label{sect:testing-clique-queries}
In this section we prove \Cref{thm:cliques-queries}.
We start by defining independent set stars and giving a container procedure to construct a sequence of inner and outer containers for every independent set star.
Then we show some basic properties of the container procedure, and prove a graph container lemma which roughly states that for every independent set star, there exists an inner and outer container from the sequence that are sufficiently small.
Finally, we use the graph container lemma to prove \Cref{thm:cliques-queries}.

\subsection{Container Procedure}
We start by defining an independent set star.
\begin{definition}
Let $G=(V,E)$ be a graph.
An \emph{independent set star} is a set $I \subseteq V,$ which we call the \emph{core}, and a set $J \subseteq V \setminus I$ such that $I$ is an independent set and none of the vertices in $J$ are adjacent to any vertex in $I.$
We write $(I,J)$ to denote the independent set star.
\end{definition}

We now give our new container procedure for independent set stars.
In particular, for an independent set star $(I,J)$ we construct an \emph{inner} container $C$ that contains $I,$ and an \emph{outer} container $D$ that contains $I \cup J.$
Our procedure only depends on the independent set $I,$ and so it can be viewed as finding an inner and an outer container such that the independent set is contained in the inner container, and the outer container contains all vertices that have no edges to the independent set (which would make them candidates to be part of an independent set star with $I$ as the core).

The construction of the inner containers is very similar to the result by Blais and Seth \cite{blaisSeth}, which itself follows closely from the original use of containers by Kleitman and Winston \cite{kleitmanWinston1982}.
Our algorithm differs primarily on the construction of the outer containers.

After initializing the fingerprint to be empty and the inner and outer containers to be $V,$ the procedure runs for $|I|/2$ iterations.
At each step of the procedure two vertices are added to the fingerprint.
The first is the vertex $u$ of $I$ with largest degree in the current inner container, and the second is the vertex $v$ of $I$ with largest degree in the current outer container, where ties are broken based on an assumed ordering of the vertices $V.$
After selecting the fingerprint vertices, all vertices that neighbour $u$ or $v$ are removed from both containers.
Further, all vertices with degree higher than $u$ in the inner container or with degree higher than $v$ in the outer container are removed from the inner container.
The algorithm repeats until all vertices from the independent set have been added to the fingerprint.

We note that the main difference between the inner and outer container in each step is that higher degree vertices are not removed from the outer container.
This is because we can only remove vertices from the outer container if they have an edge to a vertex in $I.$

\bigskip
\begin{algorithm}[!htbp]
\caption{\sc{Fingerprint \& Container Generator}}
\label{alg:FingerprintContainerQueries}
\smallskip
\KwIn{A graph $G = (V,E)$ and an independent set $I \subseteq V$ }
\medskip
Initialize $F_0 \gets \emptyset$ and $C_0,D_0 \gets V$\;

\For{$t=1,2,\dots,|I|/2$} {
    $u_t \gets$ the vertex in $I \setminus F_{t-1}$ that has the largest degree in $G[C_{t-1}].$\;

    $v_t \gets$ the vertex in $I \setminus F_{t-1} \setminus u_t$ that has the largest degree in $G[D_{t-1}]$\;

    \medskip 
    \tcp{Add $u_t$ and $v_t$ to the fingerprint}
    $F_t \gets F_{t-1} \cup \{u_t,v_t\}$\;
    
    \medskip
    \tcp{Remove all the neighbours of $u_t,v_t$ from the containers}
    $C_t \gets C_{t-1} \setminus \big\{ w \in C_{t-1} : (u_t,w) \in E \textrm{ or } (v_t,w) \in E\big\}$\;
    $D_t \gets D_{t-1} \setminus \big\{ w \in D_{t-1}: (u_t,w) \in E \textrm{ or }(v_t,w) \in E\big\}$\;

    \medskip
    \tcp{And remove all vertices from $C_t$ with higher degree than $u_t$ in $G[C_{t-1}]$ and vertices with higher degree than $v_t$ in $G[D_{t-1}]$}
    $C_t \gets C_t \setminus  \big\{ w \in C_{t-1} : \deg_{G[C_{t-1}]}(w) > \deg_{G[C_{t-1}]}(u_t) \textrm{  or  } \deg_{G[D_{t-1}]}(w) > \deg_{G[D_{t-1}]}(v_t) \big\}$\;    
}
Return $F_1,\ldots,F_{|I|/2},$ $C_1,\ldots,C_{|I|/2}$ and $D_1,\ldots, D_{|I|/2}$\;
\end{algorithm}
\bigskip

As in the previous section, we utilize the sequence of fingerprints $F_1,F_2,\ldots,F_{|I|/2}$ and their corresponding containers $C_1,C_2,\ldots,C_{|I|/2}$ and $D_1,D_2,\ldots,D_{|I|/2}$ generated by the algorithm for the independent set $I$. Further, it is also convenient to extend the definition with $C_t=F_t=I$ and $D_t=D_{|I|/2}$ for all $t>|I|/2.$

We refer to $F_t$ as the \emph{$t$-th fingerprint}, $C_t$ as the \emph{$t$-th inner container}, and $D_t$ as the \emph{$t$-th outer container} of $I$.
When we consider the fingerprints or containers of multiple independent sets, we write $F_t(I), C_t(I)$ and $D_t(I)$ to denote the $t$-th fingerprint and containers of the independent set $I$. (When the current context specifies a single independent set, we will continue to write only $F_t,C_t$ and $D_t$ to keep the notation lighter.)

\subsection{Basic Properties}
By construction, the sequence of fingerprints and containers of an independent set satisfy a number of useful properties. For instance, we have 
\[
F_1 \subseteq F_2 \subseteq \cdots \subseteq F_{|I|/2} = I \subseteq C_{|I|/2} \subseteq C_{|I|/2-1} \subseteq \cdots \subseteq C_2 \subseteq C_1
\]
so that for each $t = 1,2,\ldots,|I|/2$, $F_t \subseteq I \subseteq C_t$.
We also have that for all $t,$ the set $D_t$ contains all vertices that have no neighbours in $I$ because the only vertices removed from $D_t$ are those that are adjacent to a fingerprint vertex.
As a result, for any $t$ and any set $J \subseteq V \setminus I$ that, along with $I,$ forms an independent set star $(I,J),$ we have that $J \subseteq D_t.$ 
Finally, it is useful to note that $C_t \subseteq D_t$ for all $t$ because any time a vertex is removed from $D_t$ it is also removed from $C_t.$

The construction of the algorithm also guarantees that the $t$-th container of an independent set is the same as the $t$-th container of its $t$-th fingerprint.

\begin{proposition}
\label{prop:container-closure-queries}
For any graph $G = (V,E)$, any independent set $I$ in $G$, and any $t,$ the fingerprint $F_t(I)$ and containers $C_t(I),D_t(I)$ of $I$ satisfy 
\[
C_t\big( F_t(I) \big) = C_t(I) \textrm{ and } D_t\big( F_t(I) \big) = D_t(I).
\]
\end{proposition}

\begin{proof}
If $t> |I|/2$ then $F_t(I)=I$ and the statements hold.
Now consider the case that $t \leq |I|/2.$ 
If $u_1,v_1,\dots,u_t,v_t$ are the vertices selected in the first $t$ iterations of the container algorithm on input $I$, then the algorithm selects the same vertices and forms the same containers $C_t,D_t$ when provided with input $F_t(I)$ instead of $I$.
\end{proof}

\subsection{Container Shrinking Lemma}
\label{sect:shrinking-lemma-queries}
In order to prove our graph container lemma, we need the following lemma.
This lemma states that, under the right conditions, the outer container shrinks relatively quickly.
\begin{lemma}[Outer Container Shrinking Lemma]
\label{lem:containers-shrinking2}
Let $G = (V,E)$ be a graph on $n$ vertices that is $\epsilon$-far from the \rhoIndepSet property,  let $I$ be an independent set in $G$, and let $t < |I|/2$ be an index for which the $t$-th container $D_{t}$ of $I$ has cardinality $|D_{t}| \ge \rho n$. For any $\alpha \leq \sqrt{\epsilon}/2$, if there exists a subset $D \subseteq D_{t+1}$ of $(\rho - \alpha) n$ vertices where $G[D]$ contains at most $\frac{3}{8} \epsilon n^2$ edges and $|D \cap C_{t+1}| \geq \left(\rho-\sqrt{\epsilon}/2\right)n,$ then
\[
\left|D_{t+1} \setminus D\right| \le \left( 1 - \frac{\epsilon}{4\rho \alpha} \right) \left| D_{t} \setminus D \right|.
\]
\end{lemma}
\begin{proof}
    We first give a lower bound on the number of edges with one endpoint in $D$ and one endpoint in $D_t\setminus D.$
    Let $R$ be an arbitrary set of size $\alpha n$ selected from $D_t\setminus D.$
    Since $\alpha < \sqrt{\epsilon}/2,$ then $G[R]$ has at most $\binom{\alpha n}{2} <\epsilon n^2/8$ edges.
    Further, by the premises of the lemma statement, we have that $G[D]$ has at most $\frac{3}{8}\epsilon n^2$ edges.
    Combined with the facts that $|D \cup R| = \rho n,$ and that $G$ is $\epsilon$-far from having a $\rho n$ independent set, there are at least $\epsilon n^2/2$ edges in $G$ with one endpoint in $D$ and one endpoint in $R.$
    Since $R$ is selected arbitrarily, then there are at least $\frac{|D_t \setminus D|}{\alpha n} \cdot \epsilon n^2/2=\frac{\epsilon|D_t \setminus D|n}{2\alpha}$ edges in $G$ with one endpoint in $D$ and one endpoint in $D_t\setminus D.$

    Let $C=D \cap C_{t+1}.$
    We now claim that, of these edges, at least $\frac{\epsilon|D_t \setminus D|n}{4\alpha}$ have one endpoint in $C.$
    To prove this, first observe that there are at most $\sqrt{\epsilon} n/2$ vertices in $D \setminus C,$ and so the number of edges with one endpoint in $D \setminus C$ and one endpoint in $D_t \setminus D$ is at most $$\frac{\sqrt{\epsilon} n}{2} \cdot |D_t \setminus D| \leq \frac{\epsilon n}{4\alpha} \cdot |D_t \setminus D|,$$ where the inequality is by the fact that $\alpha \leq \sqrt{\epsilon}/2.$
    Since $D \setminus C$ and $C$ partition $D,$ then the remaining edges, of which there are at least $\frac{\epsilon n}{4\alpha} \cdot |D_t \setminus D|,$ have one endpoint in $C$ and one endpoint in $D_t \setminus D,$ which completes the proof of the claim.

    Using the above claim, there exists a vertex $v$ in $C$ with at least $\frac{\epsilon|D_t \setminus D|n}{4\alpha |C|}\geq\frac{\epsilon|D_t \setminus D|}{4\alpha\rho}$ neighbours in $D_t \setminus D.$
    Since $v \in C_{t+1} \subseteq C_t$ then $v_{t+1}$ has degree at least as large as $v,$ otherwise $v$ would be removed from $C_{t+1}.$
    And $v_{t+1}$ has no neighbours in $D$ because $D \subseteq D_{t+1},$ which means that $v_{t+1}$ has at least $\frac{\epsilon|D_t \setminus D|}{4\alpha\rho}$ neighbours in $D_t \setminus D,$ all of which are removed when constructing $D_{t+1}.$
\end{proof}

\subsection{Proof of Graph Container Lemma}
\label{sect:gcl-proof-queries}
In order to prove our graph container lemma, we need the following lemma from \cite{blaisSeth}.
It is important to note that the container $C_t$ in the lemma below refers to the container procedure defined in \cite{blaisSeth}.
However, our container procedure in \Cref{alg:FingerprintContainerQueries} contains the steps from the container procedure in \cite{blaisSeth}, and the above lemma also holds for the container $C_t$ defined in \Cref{alg:FingerprintContainerQueries}.
The proof follows exactly as in \cite{blaisSeth} and so is omitted here.

\begin{lemma}[Lemma 5 in \cite{blaisSeth}]
\label{lem:GCL-indepset}
Let $G = (V,E)$ be a graph on $n$ vertices that is $\epsilon$-far from the \rhoIndepSet property. For any independent set $I$ in $G$ there exists an index $t \le \frac{8 \rho^2}{\epsilon}  \ln(\frac{2\rho}\epsilon)$ such that the size of the $t$-th container of $I$ is bounded by
\begin{equation}
\label{eq:GCL-indepset}
|C_t| \le \left( \rho - t \cdot \frac{\epsilon}{8 \rho \ln(\frac{2\rho}{\epsilon})} \right) n
\end{equation}
and $G[C_t]$ contains at most $\frac{\epsilon n^2}{4}$ edges.
\end{lemma}

We now prove our graph container lemma.
This lemma roughly states that there exists some $t$ value such that either $C_t$ is very small, or $D_t$ is somewhat small (which also implies that $C_t$ is somewhat small by the fact that $C_t \subseteq D_t$).

\begin{lemma}[Reformulation of \Cref{lem:ISstar} \protect\footnote{The collection of fingerprints/containers in \Cref{lem:ISstar} is constructed by applying \Cref{lem:GCL-indepset-queries} to every independent set in the graph. The properties of the containers follow from the basic properties of the container procedure and \Cref{prop:container-closure-queries}.}]
\label{lem:GCL-indepset-queries} 
Let $G = (V,E)$ be a graph on $n$ vertices that is $\epsilon$-far from the \rhoIndepSet property. For any independent set $I$ in $G$ there exists an index $t \le \frac{8 \rho^2}{\epsilon}  \ln(\frac{2\rho}\epsilon)$ such that either
\begin{itemize}
    \item $t \geq \frac{4\rho \ln(2\rho/\epsilon)}{\sqrt{\epsilon}}$ and $|C_t| \le \left( \rho - t \cdot \frac{\epsilon}{8 \rho \ln(\frac{2\rho}{\epsilon})} \right) n,$ or
    \item $t < \frac{4\rho \ln(2\rho/\epsilon)}{\sqrt{\epsilon}}$ and $|D_t| \le \left( \rho - t \cdot \frac{\epsilon}{8 \rho \ln(\frac{2\rho}{\epsilon})} \right) n.$
\end{itemize}
\end{lemma}

\begin{proof}
    Throughout the proof we utilize three special $t$ values, so we start be defining these.
    \begin{itemize}
        \item Let $t_c$ be the largest index such that $|C_{t}| \leq \left( \rho - t \cdot \frac{\epsilon}{8 \rho \ln(\frac{2\rho}{\epsilon})} \right) n$ and $G[C_{t}]$ has less than $\epsilon n^2/4$ edges.
        Such an index exists by \Cref{lem:GCL-indepset}.
        \item Let $t_d$ be the smallest $t$ such that $\sum_{v \in C_{t_d}} \deg_{D_{t_d}}(v) \leq \epsilon n^2/4.$
        Such an index is guaranteed to exist since the vertices in $C_{\frac{|I|}{2}+1}=I$ have no edges to $D_{\frac{|I|}{2}+1}.$
        \item Let $t^*$ be the largest index smaller than $|I|/2$ for which $|D_{t^*}| \ge \rho n.$
    \end{itemize}
    
    Given these special $t$ values, first observe that if $t_c \geq \frac{4\rho\ln(2\rho/\epsilon)}{\sqrt{\epsilon}},$ then the conclusion of the lemma is satisfied.
    For the rest of the proof, consider the case that $t_c <\frac{4\rho\ln(2\rho/\epsilon)}{\sqrt{\epsilon}}.$
    We claim that $t_d < \frac{4\rho\ln(2\rho/\epsilon)}{\sqrt{\epsilon}}$ and $|D_{t_d}| \leq \left(\rho-t_d \cdot \frac{\epsilon}{8\rho \ln(2\rho/\epsilon)}\right)n,$ which would complete the proof of the lemma.

    We first show that $t_d < \frac{4\rho\ln(2\rho/\epsilon)}{\sqrt{\epsilon}}.$
    Let $C=C_{\frac{4\rho\ln(2\rho/\epsilon)}{\sqrt{\epsilon}}},$ and observe that $G[C]$ has less than $\epsilon n^2/4$ edges because $C \subseteq C_{t_c}.$
    Since we selected $t_c$ to be maximal, then $|C| > \left(\rho-\frac{\sqrt{\epsilon}}{2}\right)n$ because otherwise we could select a larger $t$ index.
    
    We claim that $t^* < \frac{2\rho\ln(2\rho/\epsilon)}{\sqrt{\epsilon}}.$
    In order to show this, apply \Cref{lem:containers-shrinking2} for all $t=0,1,\dots,t^*-1$ with $D=C$ to find that \[ |D_{t^*} \setminus C| \le \left(1 - \frac{\sqrt{\epsilon}}{2\rho}\right)^{t^*}n.\]
    If $t^* \geq \frac{2\rho\ln(2\rho/\epsilon)}{\sqrt{\epsilon}},$ then $|D_{t^*} \setminus C| \le \frac{\epsilon}{2\rho} n.$
    Using that $|C| < \left(\rho-\frac{\epsilon}{2\rho}\right)n,$ which follows from the fact that $G[C]$ has at most $\epsilon n^2/4$ edges and that $G$ is $\epsilon$-far from having a $\rho n$ independent set, then we find that $|D_{t^*}| = |D_{t^*} \setminus C| + |C| < \rho n,$ a contradiction.
    
    Consider now values $t = t^*+1, \dots, t_d-1,$ where if there are no $t$ values in this range then we have the desired upper bound on $t_d.$
    For each such $t$ we have that $\sum_{v \in C_{t}} \deg_{D_{t}}(v) > \epsilon n^2/4,$ and so there exists at least one vertex in $C_t$ with degree least $\frac{\epsilon}{4\rho} n$ in $G[D_t]$ since $|C_t| \leq |D_t| < \rho n.$
    This means that for all $t = t^*+1, \dots, t_d-1$ the fingerprint vertex $v_t$ has degree at least $\frac{\epsilon}{4\rho}n$ in $G[D_{t-1}]$ because otherwise there would be no vertices in $C_t$ with degree at least $\frac{\epsilon}{4\rho} n$ in $G[D_t]$, and so at least $\frac{\epsilon}{4\rho}n$ vertices are removed from $D_{t-1}$ to construct $D_t$ in the $t$-th iteration of \Cref{alg:FingerprintContainerQueries}.
    Since $|D_{t^*+1}| < \rho n,$ there can be at most $\frac{2\rho}{\sqrt{\epsilon}}$ such iterations before we reach a point that $|D_t| < \left(\rho-\frac{\sqrt{\epsilon}}{2}\right)n < |C|$.
    Hence, we find that $t_d < \frac{4\rho\ln(2\rho/\epsilon)}{\sqrt{\epsilon}}$ because otherwise we could show that $|D_{\frac{4\rho\ln(2\rho/\epsilon)}{\sqrt{\epsilon}}}| < |C| = |C_{\frac{4\rho\ln(2\rho/\epsilon)}{\sqrt{\epsilon}}}|,$ which is a contradiction.
    
    We now show that the cardinality of $D_{t_d}$ satisfies the desired bound.
    As a first step, we show that its cardinality is at most $\rho n.$
    First observe that since $t_d < \frac{4\rho\ln(2\rho/\epsilon)}{\sqrt{\epsilon}}$ then $C_{t_d} \supseteq C,$ and so $|C_{t_d}| \geq (\rho -\sqrt{\epsilon}/2)n.$
    Further, by the selection of $t_d,$ the vertices in $C_{t_d}$ are adjacent to at most $\epsilon n^2/4$ edges in $G[D_{t_d}].$ 
    As a result, $|D_{t_d}|< \left(\rho-\frac{\epsilon}{2\rho}\right) n$ because otherwise we could construct a $\rho n$ set from $C_{t_d},$ at most $\frac{\sqrt{\epsilon}}{2}n$ arbitrary other vertices from $D_{t_d},$ and at most $\frac{\epsilon}{2\rho}n$ other arbitrary vertices from $V$ to demonstrate that $G$ is not $\epsilon$-far from having a $\rho n$ independent set.
    Further, since $C_{t_d} \subseteq D_{t_d},$ we have that the number of edges in $G[D_{t_d}]$ is at most $\frac{3}{8}\epsilon n^2.$

    Define $\alpha$ such that $|D_{t_d}| = (\rho - \alpha) n$.
    We now determine an upper bound on $t_d$ as a function of $\alpha,$ which will give the upper bound on $|D_{t_d}|.$
    In particular, the setting $D=D_{t_d}$ satisfies the criteria of \Cref{lem:containers-shrinking2} for all $t=0,1,\dots,t^*-1,$ and so we find that
    \[
    |D_{t^*} \setminus D_{t_d}| \le \left(1 - \frac{\epsilon}{4\rho \alpha}\right)^{t^*}n.
    \]
    Since $|D_{t^*} \setminus D_{t_d}| \ge \alpha n$, we conclude that $t^* \le \frac{4\rho \alpha}{\epsilon} \ln(1/\alpha) \le \frac{4\rho \alpha}{\epsilon} \ln (2\rho/\epsilon),$ where the second inequality uses the lower bound $\alpha > \frac{\epsilon}{2\rho}$ coming from the fact that $|D_{t_d}|< \left(\rho-\frac{\epsilon}{2\rho}\right) n.$

    Consider now values $t = t^*+1,\dots,t_d-1.$
    Using the same reasoning as above, for each such $t$ we have that $\sum_{v \in C_{t}} \deg_{D_{t}}(v) > \epsilon n^2/4,$ and so there exists at least one vertex in $C_t$ with degree at least $\frac{\epsilon}{4\rho} n$ in $G[D_t]$ since $|C_t| \leq |D_t| < \rho n.$
    This means that for all $t = t^*+1, \dots, t_d-1$ the fingerprint vertex $v_t$ has degree at least $\frac{\epsilon}{4\rho}n$ in $G[D_{t-1}]$ because otherwise there would be no vertices in $C_t$ with degree at least $\frac{\epsilon}{4\rho} n$ in $G[D_t]$, and so at least $\frac{\epsilon}{4\rho}n$ vertices are removed from $D_{t-1}$ to construct $D_t$ in the $t$-th iteration of \Cref{alg:FingerprintContainerQueries}.
    Since $|D_{t^*+1}| < \rho n$ there can be at most $\frac{4\rho \alpha}{\epsilon}$ such iterations before we reach $|D_t| = (\rho-\alpha)n=|D_{t_d}|$.
    Therefore, $t_d \leq \frac{4\rho \alpha}{\epsilon} \ln (2\rho/\epsilon)+ \frac{4\rho \alpha}{\epsilon} < \frac{8\rho \alpha}{\epsilon} \ln (2\rho/\epsilon),$ which can be rearranged to show the desired upper bound on the size of the container $D_{t_d}$.
\end{proof}

\subsection{Proof of \autoref{thm:cliques-queries}}
\label{sect:clique-queries-proof}
We are now ready to use \Cref{lem:GCL-indepset-queries} to complete the proof of \Cref{thm:cliques-queries}. The proof of the theorem also makes use of the following form of Chernoff's bound for hypergeometric distributions (see, for example, \cite{blaisSeth}).

\begin{lemma}[Chernoff's Bound]
\label{lem:chernoff}
Let $X$ be drawn from the hypergeometric distribution $H(N,K,n)$ where $n$ elements are drawn without replacement from a population of size $N$ where $K$ elements are marked and $X$ represents the number of marked elements that are drawn.
Then for any $\vartheta \geq E[X]$,
\[
\Pr\big[X \geq \vartheta\big] \leq \exp\left(-\frac{(\vartheta-E[X])^2}{\vartheta+E[X]}\right).
\]
\end{lemma}

We are now ready to prove \Cref{thm:cliques-queries}, restated below in its precise form (with the polylogarithmic term) and for the \rhoIndepSet property. 

\newtheorem*{clique-thm-precise}{Theorem 4}
\begin{clique-thm-precise}[Precise formulation]
The query complexity of $\epsilon$-testing the \rhoIndepSet property is 
\[
\mathcal{Q}_{\rhoIndepSet}(n,\epsilon) = O\left(\frac{\rho^5}{\epsilon^{7/2}} \ln^5\left(\frac 1\epsilon \right)\right).
\]
\end{clique-thm-precise}

\begin{proof}
    Let $R,S$ be a random sets of vertices drawn uniformly at random from $V$ without replacement, of sizes $r = c_1 \cdot \frac{\rho^2}{\epsilon^{3/2}} \ln^2\left(\frac 1\epsilon \right)$ and $s = c_2 \cdot \frac{\rho^3}{\epsilon^2} \ln^3\left(\frac 1\epsilon \right)$ respectively, where $c_1,c_2$ are large enough constants.
    Since the problem is non-trivial only when $\epsilon < \frac{\rho^2}{2},$ then we have that $s > r.$

    The tester queries all pairs of vertices in $R,$ and all pairs of vertices with one in vertex in $R$ and one vertex in $S.$
    Hence, the total number of edge queries is $r^2+rs = O\left(\frac{\rho^5}{\epsilon^{7/2}} \ln^5\left(\frac 1\epsilon \right)\right).$
    The tester accepts if there exists an independent set $I$ of size $\rho r$ in $R$ and a set $J$ of size $\rho s$ in $S$ such that there are no edges between a vertex in $I$ and a vertex in $J.$
    
    If $G$ contains a $\rho n$ independent set, then $R$ contains at least $\rho r$ vertices from this independent set with probability at least $\frac12$ since the number of such vertices follows a hypergeometric distribution, and so the median of this distribution is at least $\rho s$~\cite{Neumann70,KaasBurhman1980}.
    Similarly, $S$ contains at least $\rho s$ vertices from the independent set with probability at least $\frac12,$ and so the tester accepts with probability at least $\frac14.$

    In the remainder of the proof we upper bound the probability that the tester accepts when $G$ is $\epsilon$-far from containing a $\rho n$-independent set.
    For the rest of the argument, let us call a pair of containers $C_t,D_t$ \emph{small} when they satisfy the expression in the conclusion of \Cref{lem:GCL-indepset-queries}.
    In particular, we call a pair of containers $C_t,D_t$ \emph{small} when either:
    \begin{itemize}
        \item $t \geq \frac{4\rho \ln(2\rho/\epsilon)}{\sqrt{\epsilon}}$ and $|C_t| \le \left( \rho - t \cdot \frac{\epsilon}{8 \rho \ln(\frac{2\rho}{\epsilon})} \right) n,$ or
        \item $t < \frac{4\rho \ln(2\rho/\epsilon)}{\sqrt{\epsilon}}$ and $|D_t| \le \left( \rho - t \cdot \frac{\epsilon}{8 \rho \ln(\frac{2\rho}{\epsilon})} \right) n.$
    \end{itemize}

    Let us denote the vertices of the sampled set $R$ as $u_1,u_2,\ldots,u_{r}.$
    First, by \Cref{lem:GCL-indepset-queries}, for \emph{every} independent set $I$ of size $\rho r$ in $R$, there is a $t \le \frac{8\rho^2 \ln(\frac{2\rho}{\epsilon})}{\epsilon}$ and a fingerprint $F_t \subseteq I$ of size $2t$ that defines a small pair of containers $C_t(F_t)$ and $D_t(F_t)$ such that $I \subseteq C_t(F_t)$ and $D_t$ contains all vertices in $V$ that have no edges to $I.$
    Hence, the probability that the tester accepts is at most the probability that there exists some $t \le \frac{8\rho^2 \ln(\frac{2\rho}{\epsilon})}{\epsilon}$ such that $R$ contains a set of $2t$ vertices that form the fingerprint $F_t$ of a pair of small containers $C_t,D_t,$ and $R$ contains at least $\rho r - 2t$ other vertices from $C_t$ and $S$ contains at least $\rho s$ vertices from $D_t.$
    We now upper bound this probability.

    For any $2t$ distinct indices $i_1,i_2,\ldots,i_{2t} \in [r]$, consider the event where the vertices $u_{i_1},\ldots,u_{i_{2t}}$ form the fingerprint $F_t$ of a small pair of containers $C_t,D_t$. Let $X$ denote the number of vertices among the other $r-2t$ sampled vertices that belong to $C_t,$ and let $Y$ denote the number of vertices in the $S$ sample that belong to $D_t.$ We consider two cases.
    
    If $t \geq \frac{4\rho\ln(2\rho/\epsilon)}{\sqrt{\epsilon}},$ then we find that the expected value of $X$ is
    \[
    \E[X] \leq \left(\rho-\frac{t\epsilon}{8\rho\ln(\frac{2\rho}{\epsilon})}\right)(r-2t) <  \rho r - \frac{t\epsilon r}{8\rho\ln(\frac{2\rho}{\epsilon})}
    \le \rho r - 2t - \frac{t\epsilon r}{16 \rho\ln(\frac{2\rho}{\epsilon})},
    \]
    where the second inequality uses that $r-2t < r,$ and the last inequality uses the fact that $r \geq \frac{32 \rho}{\epsilon}\ln(\frac{2\rho}{\epsilon})$ because $\epsilon < \frac{\rho^2}{2}$ and $c_1$ is a large enough constant.

    By the Chernoff bound in \Cref{lem:chernoff}, the probability that we draw at least $\rho r - 2t$ vertices from $C_t$ in the final $r - 2t$ vertices drawn to form $R$ is
    \[
    \Pr[ X \ge \rho r - 2t] 
    \le \exp\left(\frac{-(\rho r-2t-E[X])^2}{\rho r - 2t+E[X]}\right)
    \leq \exp\left(-\frac{t^2\epsilon^2r}{512 \rho^3\ln^2(\frac{2\rho}{\epsilon})} \right).
    \]

    Next, if $t < \frac{4\rho\ln(2\rho/\epsilon)}{\sqrt{\epsilon}},$ then we find that the expected value of $Y$ is
    \[
    \E[Y] \leq \left(\rho-\frac{t\epsilon}{8\rho\ln(\frac{2\rho}{\epsilon})}\right)s =  \rho s - \frac{t\epsilon s}{8\rho\ln(\frac{2\rho}{\epsilon})}.
    \]

    By the Chernoff bound in \Cref{lem:chernoff}, the probability that we draw at least $\rho s$ vertices from $D_t$ in the $s$ vertices drawn to form $S$ is
    \[
    \Pr[ Y \ge \rho s] 
    \le \exp\left(\frac{-(\rho s-E[X])^2}{\rho s +E[X]}\right)
    \leq \exp\left(-\frac{t^2\epsilon^2s}{64 \rho^3\ln^2(\frac{2\rho}{\epsilon})} \right).
    \]

    Therefore, by applying a union bound over all values $t \le \frac{8\rho^2 \ln(\frac{2\rho}{\epsilon})}{\epsilon}$ and all possible choices of $2t$ indices from $[r]$, the probability that \emph{any} set of at most $\frac{16\rho^2 \ln(\frac{2\rho}{\epsilon})}{\epsilon}$ vertices in $R$ form the fingerprint of a small pair of containers from which we sample at least $\rho r$ vertices in $R$ and $\rho s$ vertices in $S$ is at most

    \[
    \sum_{t=1}^{4\rho \ln(2\rho/\epsilon)/\sqrt{\epsilon}} \binom{r}{2t} \exp\left(-\frac{t^2\epsilon^2s}{64 \rho^3\ln^2(\frac{2\rho}{\epsilon})} \right) + \sum_{t=4\rho \ln(2\rho/\epsilon)/\sqrt{\epsilon}}^{8\rho^2 \ln(2\rho/\epsilon)/\epsilon} \binom{r}{2t} \exp\left(-\frac{t^2\epsilon^2r}{512 \rho^3\ln^2(\frac{2\rho}{\epsilon})} \right)
    \]
    \[
    \le \sum_{t=1}^{4\rho \ln(2\rho/\epsilon)/\sqrt{\epsilon}} \exp\left( 2t \ln r - \frac{t^2\epsilon^2s}{64 \rho^3\ln^2(\frac{2\rho}{\epsilon})} \right) +  \sum_{t=4\rho \ln(2\rho/\epsilon)/\sqrt{\epsilon}}^{8\rho^2 \ln(2\rho/\epsilon)/\epsilon}\exp\left( 2t \ln r - \frac{t^2\epsilon^2r}{512 \rho^3\ln^2(\frac{2\rho}{\epsilon})} \right)
    \]
    where the inequality uses the upper bound $\binom{r}{2t} \le r^{2t}.$
    Finally, using that $r = c_1 \cdot \frac{\rho^2}{\epsilon^{3/2}} \ln^2(\frac{1}{\epsilon})$ and $s = c_2 \cdot \frac{\rho^3}{\epsilon^2} \ln^3(\frac 1\epsilon)$ for large enough constants $c_1,c_2,$ and the fact that $t \geq \frac{4\rho}{\sqrt{\epsilon}}\ln(\frac{2\rho}{\epsilon})$ for the summation on the right, the above expression can be upper bounded by $o(1).$
    Hence, the probability that the tester accepts when $G$ is $\epsilon$-far from having a $\rho n$ independent set is $o(1).$
\end{proof}

\bibliographystyle{alpha}
\bibliography{references}

\begin{thebibliography}{GGR98}

\bibitem[AG11]{AvigadGoldreich11}
Lidor Avigad and Oded Goldreich.
\newblock Testing graph blow-up.
\newblock In {\em International Workshop on Approximation Algorithms for
  Combinatorial Optimization}, pages 389--399. Springer, 2011.

\bibitem[AS03]{alon2003testing}
Noga Alon and Asaf Shapira.
\newblock Testing satisfiability.
\newblock {\em Journal of Algorithms}, 47(2):87--103, 2003.

\bibitem[BMS15]{baloghIndependentSetsHypergraphs2015}
J{\'o}zsef Balogh, Robert Morris, and Wojciech Samotij.
\newblock Independent sets in hypergraphs.
\newblock {\em Journal of the American Mathematical Society}, 28(3):669--709,
  2015.

\bibitem[BMS18]{baloghMorrisSamotijContainersSurvey2018}
J{\'o}zsef Balogh, Robert Morris, and Wojciech Samotij.
\newblock The method of hypergraph containers.
\newblock In {\em Proceedings of the International Congress of Mathematicians:
  Rio de Janeiro 2018}, pages 3059--3092. World Scientific, 2018.

\bibitem[BS23]{blaisSeth}
Eric Blais and Cameron Seth.
\newblock Testing graph properties with the container method.
\newblock In {\em to appear in 2023 IEEE 64rd Annual Symposium on Foundations
  of Computer Science (FOCS)}. IEEE, 2023.

\bibitem[CS01]{czumaj2001testing}
Artur Czumaj and Christian Sohler.
\newblock Testing hypergraph coloring.
\newblock In {\em Automata, Languages and Programming: 28th International
  Colloquium, ICALP 2001 Crete, Greece, July 8--12, 2001 Proceedings 28}, pages
  493--505. Springer, 2001.

\bibitem[FLS04]{feigeCliqueTesting2004}
Uriel Feige, Michael Langberg, and Gideon Schechtman.
\newblock Graphs with tiny vector chromatic numbers and huge chromatic numbers.
\newblock {\em SIAM Journal on Computing}, 33(6):1338--1368, 2004.

\bibitem[FR21]{fiat2021efficient}
Nimrod Fiat and Dana Ron.
\newblock On efficient distance approximation for graph properties.
\newblock In {\em Proceedings of the 2021 ACM-SIAM Symposium on Discrete
  Algorithms (SODA)}, pages 1618--1637. SIAM, 2021.

\bibitem[GGR98]{goldreichPropertyTesting1998}
Oded Goldreich, Shafi Goldwasser, and Dana Ron.
\newblock Property testing and its connection to learning and approximation.
\newblock {\em Journal of the ACM}, 45(4):653--750, 1998.

\bibitem[GR11]{GoldreichRon11}
Oded Goldreich and Dana Ron.
\newblock Algorithmic aspects of property testing in the dense graphs model.
\newblock {\em SIAM Journal on Computing}, 40(2):376--445, 2011.

\bibitem[GT03]{GoldreichTrevisan03}
Oded Goldreich and Luca Trevisan.
\newblock Three theorems regarding testing graph properties.
\newblock {\em Random Structures \& Algorithms}, 23(1):23--57, 2003.

\bibitem[JPP23]{JenssenPP23}
Matthew Jenssen, Will Perkins, and Aditya Potukuchi.
\newblock Approximately counting independent sets in bipartite graphs via graph
  containers.
\newblock {\em Random Structures \& Algorithms}, 2023.

\bibitem[KB80]{KaasBurhman1980}
Rob Kaas and Jan~M Buhrman.
\newblock Mean, median and mode in binomial distributions.
\newblock {\em Statistica Neerlandica}, 34(1):13--18, 1980.

\bibitem[KW82]{kleitmanWinston1982}
Daniel~J. Kleitman and Kenneth~J. Winston.
\newblock On the number of graphs without 4-cycles.
\newblock {\em Discrete Mathematics}, 41(2):167--172, 1982.

\bibitem[Neu70]{Neumann70}
Peter Neumann.
\newblock {\"U}ber den median einiger dikreter verteilungen und eine damit
  zusammenh{\"a}ngende monotone konvergenz.
\newblock {\em Wissenschaftliche Zeitschrift der Technischen Universit{\"a}t
  Dresden}, 19:29--33, 1970.

\bibitem[NR18]{NakarRon2018}
Yonatan Nakar and Dana Ron.
\newblock On the testability of graph partition properties.
\newblock In {\em Approximation, Randomization, and Combinatorial Optimization.
  Algorithms and Techniques ({APPROX/RANDOM} 2018)}, page~13, 2018.

\bibitem[Sap05]{Sapozhenko05}
Alexander Sapozhenko.
\newblock Systems of containers and enumeration problems.
\newblock In {\em Stochastic Algorithms: Foundations and Applications (SAGA
  2005)}, pages 1--13. Springer, 2005.

\bibitem[Soh12]{Sohler12}
Christian Sohler.
\newblock Almost optimal canonical property testers for satisfiability.
\newblock In {\em Proceedings of the 53rd Annual {IEEE} Symposium on
  Foundations of Computer Science}, pages 541--550, 2012.

\bibitem[ST15]{saxtonThomasonHypergraphContainers2015}
David Saxton and Andrew Thomason.
\newblock Hypergraph containers.
\newblock {\em Inventiones mathematicae}, 201(3):925--992, 2015.

\bibitem[Zam23]{Zamir23}
Or~Zamir.
\newblock Algorithmic applications of hypergraph and partition containers.
\newblock In {\em Proceedings of the 55th Annual ACM Symposium on Theory of
  Computing}, 2023.

\end{thebibliography}

\end{document}